\documentclass[times,sort&compress,3p]{elsarticle}

\usepackage{hyperref}
\usepackage{algorithmic}
\usepackage{mathrsfs}
\usepackage{graphicx, setspace, amsmath, amssymb, subfigure, url, multirow, booktabs, verbatim, bm,threeparttable, amsthm, color}%,xcolor}
\usepackage[linesnumbered,boxed,ruled,commentsnumbered]{algorithm2e}

\newtheorem{theorem}{Theorem}

\newtheorem{condition}{Condition}
\newtheorem{lemma}{Lemma}

\begin{document}

\begin{frontmatter}

\title{Hierarchical False Discovery Rate Control for High-dimensional Survival Analysis with Interactions}

%\tnotetext[mytitlenote]{Fully documented templates are available in the elsarticle package on \href{http://www.ctan.org/tex-archive/macros/latex/contrib/elsarticle}{CTAN}.}

%% Group authors per affiliation:
\author[Yale]{Weijuan Liang}
\ead{weijuan.liang@yale.edu}

\author[XMU]{Qingzhao Zhang\corref{mycorrespondingauthor}}
\cortext[mycorrespondingauthor]{Corresponding author}
\ead{qzzhang@xmu.edu.cn}

\author[Yale]{Shuangge Ma\corref{mycorrespondingauthor}}
\ead{shuangge.ma@yale.edu}

\address[Yale]{Department of Biostatistics, Yale School of Public Health, New Haven, Connecticut, USA}
%\address[RUC]{School of Statistics, Renmin University of China, Beijing, China}
\address[XMU]{Department of Statistics and Data Science, School of Economics, The Wang Yanan Institute for Studies in Economics, and Fujian Key Lab of Statistics, Xiamen University, Xiamen, China}

%% Group authors per affiliation:
%\author{Elsevier\fnref{myfootnote}}
%\address{Radarweg 29, Amsterdam}
%\fntext[myfootnote]{Since 1880.}

%% or include affiliations in footnotes:
%\author[mymainaddress,mysecondaryaddress]{Elsevier Inc}
%\ead[url]{www.elsevier.com}

%\author[mysecondaryaddress]{Global Customer Service\corref{mycorrespondingauthor}}
%\cortext[mycorrespondingauthor]{Corresponding author}
%\ead{support@elsevier.com}

%\address[mymainaddress]{1600 John F Kennedy Boulevard, Philadelphia}
%\address[mysecondaryaddress]{360 Park Avenue South, New York}

\begin{abstract}
With the development of data collection techniques, analysis with a survival response and high-dimensional covariates has become routine. Here we consider an interaction model, which includes a set of low-dimensional covariates, a set of high-dimensional covariates, and their interactions. This model has been motivated by gene-environment (G-E) interaction analysis, where the E variables have a low dimension, and the G variables have a high dimension. For such a model, there has been extensive research on estimation and variable selection. Comparatively, inference studies with a valid false discovery rate (FDR) control have been very limited. The existing high-dimensional inference tools cannot be directly applied to interaction models, as interactions and main effects are not ``equal". In this article, for high-dimensional survival analysis with interactions, we model survival using the Accelerated Failure Time (AFT) model and adopt a ``weighted least squares + debiased Lasso'' approach for estimation and selection. A hierarchical FDR control approach is developed for inference and respect of the ``main effects, interactions'' hierarchy. { The asymptotic distribution properties of the debiased Lasso estimators} are rigorously established. Simulation demonstrates the satisfactory performance of the proposed approach, and the analysis of a breast cancer dataset further establishes its practical utility. 

%This template helps you to create a properly formatted \LaTeX\ manuscript.
\end{abstract}

\begin{keyword}
High-dimensional survival; Interaction analysis; FDR control; Hierarchy; Debiased Lasso.
%\texttt{elsarticle.cls}\sep \LaTeX\sep Elsevier \sep template
%\MSC[2010] 00-01\sep  99-00
\end{keyword}

\end{frontmatter}

%\linenumbers

\section{Introduction}
With the development of data collection techniques, analysis with a censored survival outcome and high-dimensional covariates has become routine. As a representative example, in cancer research, modeling the relationship between overall, disease-specific, and progression-free survival and genetic variables is of critical significance and has been extensively examined \citep{he2022rank, chai2019inference, vansteelandt2022assumption}. Here we further focus on an interaction model that includes a set of low-dimensional covariates, a set of high-dimensional covariates, and their interactions. A strong motivation for this model is gene-environment (G-E) interaction analysis \citep{wu2020structured,ren2023robust}, under which the E variables have a low dimension, and the G variables have a high dimension. It is noted that there are many other scientific domains and interaction models that have similar properties \citep{haris2016convex, du2021lasso,wang2022two}. For high-dimensional survival analysis with interactions, in recent years, there have been extensive works on estimation and variable selection \citep{wu2018dissecting, wu2020structured}. Among the available techniques, penalization has been especially popular \citep{wu2018dissecting}. In published interaction analysis studies especially the recent ones, respecting the ``main effects, interactions'' hierarchy has been strongly advocated \citep{wu2018dissecting, xu2022multidimensional}. Under this hierarchy, if an interaction term is selected as important, then at least one of its corresponding main effects (under the weak hierarchy) or both main effects (under the strong hierarchy) need to be selected. Special estimation and selection procedures need to be developed to satisfy this hierarchy \citep{bien2013lasso, wu2020structured}. For the specific interaction model considered in this study, the hierarchy reduces to that, if an interaction term is selected as important, then the corresponding main effect from the high-dimensional covariates needs to be selected. This has been well established in the G-E interaction analysis \citep{wu2020structured} and other studies. 

Compared to estimation and selection, inference, which is needed to draw more definitive conclusions, can be more challenging, and research remains limited. For high-dimensional analysis without interactions, a few methods have been developed, including the debiased Lasso method \citep{javanmard2014confidence}, data splitting method \citep{dai2022false}, model-X knockoff filter method \citep{candes2018panning}, and others. Our literature review suggests that there have been many more works for continuous and categorical outcomes under generalized linear models than for censored survival outcomes. Directly applying these methods to interaction models may lead to a violation of the ``main effects, interactions'' hierarchy. There has been some effort on hierarchical testing, under which hypotheses in the lower hierarchy are tested only when their counterparts in the higher hierarchy are rejected. Here it is noted that the concept of hierarchical testing goes beyond interaction models. Examples of hierarchical testing include \cite{yekutieli2008hierarchical}, which develops a general framework for testing hierarchically ordered hypotheses with a specific false discovery rate (FDR) control under the independence assumption. Lynch et al. \cite{lynch2017control} develops a hierarchical testing procedure for a fixed sequence structure where the testing order of hypotheses is known a priori, and both arbitrary and negative dependencies are considered. Bogomolov et al. \cite{bogomolov2021hypotheses} introduces a multiple hierarchical testing procedure using p-values that controls global error rates in a tree structure. G'Sell et al. \cite{g2016sequential} develops two sequential selection procedures, namely ForwardStop and StrongStop, and proves that these procedures control FDR at a pre-specific level while maintaining the ordering of testings. These hierarchical or sequential FDR control procedures require valid p-values, which are highly nontrivial in high-dimensional interaction analysis, especially with a censored survival outcome.

In this article, to fill an important knowledge gap, we consider inference for high-dimensional survival analysis with interactions. For modeling survival, we adopt the Accelerated Failure Time (AFT) model, which can be preferred over the Cox model under high-dimensional settings with its computational simplicity and simple interpretations \citep{wei1992accelerated, chai2019inference}. For regularized estimation, we adopt the debiased Lasso technique, which has served as the basis for multiple inference studies \citep{javanmard2014confidence, van2014asymptotically, javanmard2019false}, although we note that it has not been applied to the present data/model setting. We adopt a two-step inference procedure that ensures the ``main effects, interactions'' hierarchy. Specifically, in the first step, the high-dimensional main effects are tested. In the second step, interactions are tested only when their corresponding main effect hypotheses are rejected. An explicit computational procedure is developed, and we rigorously prove that it can control FDR at a pre-specified level. This study nontrivially advances from \cite{wu2020structured, wu2018dissecting} by conducting inference, from \cite{javanmard2014confidence, javanmard2019false} by analyzing survival data with high-dimensional interactions, and from \cite{yekutieli2008hierarchical, lynch2017control} by not assuming pre-existing valid p-values. Beyond the methodological and theoretical advancements, the simulation and data analysis also show that this study can provide a practically useful tool. 

The rest of the article is organized as follows. In Section \ref{sec:2}, we describe the data/model setup and the debiased Lasso estimator and develop the hierarchical FDR control procedure. In Section \ref{sec:3}, we establish asymptotic normality of the debiased Lasso estimator and provide the theoretical guarantee for the proposed approach. Extensive simulations and comparisons with six alternatives are conducted in Section \ref{sec:4} to gauge practical performance. In Section \ref{sec:5}, we analyze a breast cancer dataset to demonstrate practical utility. The proofs and additional numerical results are provided in Appendix.

\section{Methodology}
\label{sec:2}
\subsection{Notations}
Denote $\text{1}(\cdot)$ as the indicator function and $[d] = \{1,\cdots, d\}$ as the set of the first $d$ integers. For $m$-dimensional vector $\bm{a}$ with the $j$-th entry $a_j$, its ${l}_p$ and $l_\infty$ norms are $\|\bm{a}\|_p = (\sum_{j=1}^m |a_j|^p)^{1/p}$ and $\|\bm{a}\|_\infty = \max_{j\in[m]} |a_j|$, respectively.
Let $\bm{a}_\mathcal{S}$ denote the sub-vector of $\bm{a}$ with the entry indices restricted to set $\mathcal{S}$. Let $\mathbf{I}$ generically denote an identity matrix, and its dimension will be clarified when needed. For matrix $\mathbf{B}$, let $\bm{b}_{i}$ be the transpose of its $i$-th row and $\bm{b}_{,j}$ be its $j$-th column. Its $l_1$ and $l_\infty$ norms are $\|\mathbf{B} \|_{1} =\sup_j{\sum_i |b_{ij}|}$ and $\|\mathbf{B}\|_\infty = \sup_{i} \sum_{j}|b_{ij}|$, respectively. 
The maximum and minimum eigenvalues of matrix $\mathbf{B}$ are denoted by $\sigma_{\max}(\mathbf{B})$ and $\sigma_{\min}(\mathbf{B})$, respectively. For set $\mathcal{A}$, denote its cardinality and complementary set as $|\mathcal{A}|$ and $\mathcal{A}^c$, respectively. A random variable $\varepsilon$ is sub-Gaussian, if there exist some positive constants $k_1, k_2$ such that the tail probability of $\varepsilon$ satisfies $P(|\varepsilon|>t) \leq k_2 \exp(-k_1 t^2)$ for all $t\geq 0$. $\mathcal{N}(\mu, \sigma^2)$ denotes a normal distribution with mean $\mu$ and variance $\sigma^2$.

\subsection{Modeling and hierarchical FDR control}

Denote $T$ and $C$ as the logarithm of event and censoring times, respectively. Assume that $T$ and $C$ are independent. Let $X=(X_1, \cdots, X_d)^\top$ be the $d$-dimensional covariates with $n\ll d$, and $Z=(Z_1, \cdots, Z_q)^\top$ be the $q$-dimensional covariates with a fixed $q$. Without loss of generality, assume that the response and covariates are properly centered such that the intercept term can be omitted. Under the AFT model,
\begin{equation}
\label{eq:2.1}
T=  \sum_{j=1}^d X_j \alpha_{0,j} + \sum_{k=1}^q Z_k \gamma_{0,k} + \sum_{j=1}^d \sum_{k=1}^q X_j Z_k \beta_{0,jk}+\varepsilon,
\end{equation}
where $\{\alpha_{0,j}\}_{1}^d$, $\{\gamma_{0,k}\}_{1}^q$, and $\{\beta_{0,jk}\}_1^{d,q}$ are the true regression coefficients. $\varepsilon$ is the random error and satisfies $E(\varepsilon|X,Z) =0$. The distribution of $\varepsilon$ is otherwise unspecified, making the modeling flexible. 
Denote the censoring indicator as $\delta=\text{1} (T \leq C)$.
Consider the scenario where $\{\alpha_{0,j}\}_{1}^d$ and $\{\beta_{0,jk}\}_1^{d,q}$ are sparse while $\{\gamma_{0,k}\}_{1}^q$ is not.

%Define the rejected main effect index set as
%$\mathcal{A}_1=\left\{j\in [d]: j ~\text{such that the null hypothesis}~ \mathcal{H}_j^0: \alpha_{0,j} = 0~\text{is rejected}\right\}.$
%&\mathcal{A}_{3} = \left\{(j,k): j\in \mathcal{A}_1, k \in [q]~ \text{such that the null hypothesis}~ H_{jk}^0: \beta_{0,jk} = 0~ \text{is rejected} \right\}.
%\end{aligned}
%\end{equation*}

We rewrite model (\ref{eq:2.1}). Denote the augmented covariate  vector as $\Phi = (X^\top, Z^\top, (X\otimes Z)^\top)^\top$ and its dimension as $p=d+(d+1)q $. Its associated $p$-dimensional coefficient vector is denoted as $\bm{\theta}_0 = (\theta_{0,1}, \cdots, \theta_{0,p})^\top$ with
\begin{align*} 
\begin{split} 
\theta_{0,l}= \left \{ 
\begin{array}{ll} 
    \alpha_{0,j},                    &\text{for}~ l=j~\text{and}~j\in[d],\\ 
    \gamma_{0,k},     & \text{for}~l = d+k ~\text{and}~ k\in[q],  \\ 
    \beta_{0,jk},                                 &  \text{for}~ l= d+ jq+k, j \in[d] ~\text{and}~ k \in[q],
\end{array} 
\right. 
\end{split} 
\end{align*}
where $\otimes$ denotes the Kronecker product. Then model (\ref{eq:2.1}) can be written as:
\begin{equation}
\label{eq:2.2}
T= \Phi^\top \bm{\theta}_0 + \varepsilon.
\end{equation}
Denote the truly important covariate index set as $\mathcal{S}=\{j \in [p] : \theta_{0,j} \neq 0\}$. The set of the rejected main effects and the set of the rejected interactions associated with the $j$-th high-dimensional covariate, respectively, have index sets: 
\begin{align*} 
&\mathcal{A}_1=\left\{j \in [d]: \text{the null hypothesis}~ \mathcal{H}_j^0: \theta_{0,j} = 0~\text{is rejected}\right\},~ \text{and}\\
&\mathcal{A}_{2j} =\{k \in \mathcal{B}_{2j}:  \text{the null hypothesis}~ \mathcal{H}_k^0: \theta_{0,k} = 0~ \text{is rejected}\},
\end{align*}
where $\mathcal{B}_{2j} = \{d+jq+1,\cdots, d+jq+q\}$ is the index set of all the interactions corresponding to the $j$-th high-dimensional main effect. 
Denote the selected important covariate index set as $\widehat{\mathcal{S}}$. A false discovery occurs if $j\in \mathcal{S}^c \cap \widehat{\mathcal{S}}$.

In our analysis, the goal is to identify the rejected hypothesis index set $\widehat{\mathcal{S}} = \mathcal{A}_1\cup (\bigcup_{j\in \mathcal{A}_1} \mathcal{A}_{2j})$ that can control the overall FDR at a pre-specified level $\alpha$. Here we note that, in a few studies, FDR has been controlled for main effects and interactions separately, while our analysis is consistent with those that controls the overall FDR. Under the ``main effects, interactions'' hierarchy, a hypothesis on an interaction effect is tested only when its associated hypothesis on the corresponding high-dimensional main effect is rejected. 
Define FDR and false discovery proportion (FDP) for the main effects, $\text{FDR}_0$ and $\text{FDP}_0$, respectively, as:
$$
\text{FDR}_0 = E\left\{ 1(R>0)\cdot \frac{|\mathcal{A}_1\cap \mathcal{S}^c|}{R} \right\} = E[\text{FDP}_0], ~~ \text{FDP}_0=\frac{|\mathcal{A}_1\cap \mathcal{S}^c|}{ \max (|\widehat{\mathcal{S}}|, 1)},
$$
where $R = |\widehat{\mathcal{S}}|$ is the total number of discoveries, and $|\mathcal{A}_1\cap \mathcal{S}^c| = |\{j \in {\mathcal{A}_1}: \theta_{0,j} = 0 \}|$ is the total number of main effect false discoveries.
For $j \in [d]$, define FDR and FDP for the set of interactions corresponding to the $j$-th high-dimensional main effect, FDR$_j$ and FDP$_j$, respectively, as: 
    \begin{align*}
        \text{FDR}_j = E\left\{ 1( j \in \mathcal{A}_1, R>0) \cdot  \frac{|\mathcal{A}_{2j} \cap \mathcal{S}^c|}{R} \right\} = E[\text{FDP}_j], \quad
        \text{FDP}_j = \frac{|\mathcal{A}_{2j} \cap \mathcal{S}^c|}{ \max (|\widehat{\mathcal{S}}|, 1)}, 
    \end{align*}
where $|\mathcal{A}_{2j} \cap \mathcal{S}^c| = | \{k \in {\mathcal{A}_{2j}}: \theta_{0,k} = 0, j \in\mathcal{A}_1 \}|$ is the the number of false discoveries for all the interactions associated with the $j$-th high dimensional main effect. Note that if $|\widehat{\mathcal{S}}|=0$, we define FDP$_j$=0 for all $j=0,\cdots,d$.
The overall FDR is defined as
$$
\text{FDR} = \sum_{j=0}^d \text{FDR}_j.
$$

\subsection{Debiased Lasso weighted least squared estimation}

Suppose that there are $n$ iid subjects. We observe $\{(y_i, \delta_i, \bm{x}_i, \bm{z}_i)\}_1^n$, where $y_i=\text{min}(T_i, C_i)$, and $\delta_i$,  $\bm{x}_i  \in \mathbb{R}^{d\times 1}$ and $\bm{z}_i \in \mathbb{R}^{q \times 1}$ are the $i$-th realization of $\delta$, $X$ and $Z$, respectively. Denote $\mathbf{\Phi} \in \mathbb{R}^{n\times p}$ as the augmented design matrix with the $i$-th row's transpose being $\bm{\phi}_i = (\bm{x}_i^\top, \bm{z}_i^\top, (\bm{x}_i \otimes \bm{z}_i)^\top)^\top$. 
For estimation, we consider the computationally highly advantageous weighted least squared approach \citep{stute1993consistent}. Assume that the observations are sorted according to $y_i$'s. The Kaplan-Meier weights are:
$$w_1=\frac{\delta_1}{n}, ~ w_i = \frac{\delta_i}{n-i+1}\prod \limits_{j = 1}^{i-1}\left( \frac{n-j}{n-j+1}\right)^{\delta_j}, ~ i = 2, \cdots,n,$$
and the rescaled weight matrix is  $\mathbf{W}= \text{diag} \{nw_1,\cdots, nw_n\} \in \mathbb{R}^{n\times n}$. The Lasso penalized estimator $\widehat{\bm{\theta}}$ is defined as: 
\begin{equation}
\label{eq:2.3}
\widehat{\bm{\theta}}= \mathop{\arg\min} \limits_{\bm{\theta} \in \mathbb{R}^p} \left\{ \frac{1}{2n} \| \mathbf{W}^{1/2} (\mathbf{y}-\bm{\Phi} \bm{\theta}) \|_2^2 + \lambda \| \bm{\theta}\|_1\right\},
\end{equation}
where $\mathbf{y} = (y_1, \cdots, y_n)^\top$ and $\lambda >0$ is a data-dependent tuning parameter.
Under low-dimensional settings and without penalization, Stute \cite{stute1996distributional} proves consistency and asymptotic normality of the estimator as $n \rightarrow \infty$ for a fixed $p$ under mild conditions. With the Lasso penalty, the estimator is expected to be biased. To remove bias and generate a better-behaved estimator, we consider the debiased Lasso approach with estimator: 
\begin{equation}
\label{eq:2.4}
\widehat{\bm{\theta}}^d = \widehat{\bm{\theta}} +\frac{1}{n}\widehat{\mathbf{M}}\bm{\Phi}^{\top}  \mathbf{W} (\mathbf{y}- \bm{\Phi}\widehat{\bm{\theta}}).
\end{equation}
Here, $\widehat{\mathbf{M}} \in \mathbb{R}^{p\times p}$ is the ``decorrelating" matrix \citep{javanmard2014confidence, van2014asymptotically, javanmard2019false}. 
Following  \cite{javanmard2014confidence, javanmard2019false}, we construct $\widehat{\mathbf{M}}$  by solving the transpose of its $i$-th row $\widehat{\bm m}_i$ from the following convex problem:
\begin{equation}
\begin{aligned}
\label{eq:2.5}
\min \quad \quad \quad \quad & \bm{m}^\top \widehat{\bm{\Gamma}}\bm{m}, \\
 \text{subject to } \quad &\| \widehat{\bm{\Gamma}} \bm{m} - \bm{e}_i \|_{\infty} \leq \mu,\\
&\|\mathbf{W}^{1/2} \mathbf{\Phi} \bm{m}\|_\infty \leq n^{c_0}, 
\end{aligned}
\end{equation}
for an arbitrary fixed $1/4 < c_0 <1/2$,
where $\widehat{\mathbf{\Gamma}} = \mathbf{\Phi}^\top \mathbf{W} \mathbf{\Phi}/n$ is the empirical weighted covariance matrix, $\bm{e}_i$ is a $p$-dimensional vector with the $i$-th entry being 1 and the rest being 0, and $\mu$ controls the entry-wise $l_\infty$ norm of $\widehat{\mathbf{M}}\widehat{\mathbf{\Gamma}}-\mathbf{I}$ and the bias of $\widehat{\bm{\theta}}^d$ as shown in Theorem \ref{theorem1}. As the construction of $\widehat{\mathbf{M}}$ demonstrates, debiasing uses a reasonable ``relaxed" alternative to approximate the inverse of $\hat{\mathbf{\Gamma}}$.

To set the stage to derive asymptotic normality of  $\sqrt{n}(\widehat{\bm{\theta}}^d -\bm{\theta}_0)$, we plug (\ref{eq:2.2}) into (\ref{eq:2.4}) and obtain:
\begin{equation}
\label{eq:2.6}
\sqrt{n}(\widehat{\bm{\theta}}^d -\bm{\theta}_0) = \frac{1}{\sqrt{n}} \widehat{\mathbf{M}}\bm{\Phi}^{\top}  \mathbf{W} \bm{\varepsilon}-\sqrt{n}(\widehat{\mathbf{M}} \widehat{\bm{\Gamma}} - \mathbf{I})(\widehat{\bm{\theta}}-\bm{\theta}_{0}),
\end{equation}
where $\mathbf{I} \in \mathbb{R}^{p\times p}$ is an identity matrix. Denote $\bm{v}=\widehat{\mathbf{M}}\bm{\Phi}^{\top}  \mathbf{W} \bm{\varepsilon} /\sqrt{n}$ and $\bm{\Delta} = \sqrt{n}(\widehat{\mathbf{M}} \widehat{\bm{\Gamma}} - \mathbf{I})(\widehat{\bm{\theta}}-\bm{\theta}_{0})$. In the next section, we prove that $\bm{v}$ is entry-wise asymptotically normal in Theorem \ref{theorem1} (a) and that $\bm{\Delta}$ is asymptotically negligible in Theorem \ref{theorem1} (b), under some commonly assumed conditions. The asymptotic covariance of $\bm{v}$ can be estimated based on the observed data. Specifically, define:
$$
\begin{aligned}
&\widehat{\varphi}_{j}\left(\bm\phi_{i}, y_{i}\right)=\phi_{ij}\left(y_{i}- \bm\phi_{i}^{\top} \widehat{\bm\theta} \right),\\
&\widehat{\tau}_{0}(y)=\exp \left[\sum_{{k: y_{k}<y,\delta_{k}=0}} \frac{1}{n-\sum_{l=1}^{n} 1(y_l \leq y_{k})}\right],\\
&\widehat{\tau}_{1}^{j}\left(y\right)=\sum_{{k: y_{k}>y, \delta_k=1}} \left[ \frac{\widehat{\varphi}_{j}\left(\bm\phi_{k}, y_{k}\right) \widehat{\tau}_{0}\left(y_{k}\right)}{n-\sum_{l=1}^{n} 1(y_{l} \leq y)} \right],\\
&\widehat{\tau}_{2}^{j}\left(y\right)=\sum_{k: y_{k}<y, \delta_k = 0}\left\{ \frac{\sum_{l: y_{l}>y_{k}, \delta_{l}=1} \widehat{\varphi}_{j}\left(\bm\phi_{l}, y_{l}\right) \widehat{\tau}_{0}\left(y_{l}\right)}{\left[n-\sum_{m=1}^{n} 1(y_{m} \leq y_{k}) \right]^{2}}\right\}.
\end{aligned}
$$
For the $i$-th sample, let
$\widehat{\zeta_{j}}\left(\bm\phi_{i}, y_{i}\right)=\widehat{\varphi}_{j}\left(\bm\phi_{i}, y_{i}\right) \widehat{\tau}_{0}\left(y_{i}\right) \delta_i+\widehat{\tau}_{1}^{j}\left(y_{i}\right)\left(1-\delta_{i}\right)-\widehat{\tau}_{2}^{j}\left(y_{i}\right)$, $\widehat{\sigma}_{jk}$ be the sample covariance of $\widehat{\zeta}_{j}$ and $\widehat{\zeta}_{k}$,
and $\widehat{\bm\Sigma} = (\widehat{\sigma}_{jk})_{p\times p}$.

\subsection{Hierarchical FDR control}
Define $\widehat{\bm\Lambda}=\widehat{\mathbf{M}} \widehat{\bm{\Sigma}}\widehat{\mathbf{M}}^\top$ and its population counterpart as $\bm{\Lambda}$.
Theorem \ref{theorem1} (a) below establishes that 
$$\frac{v_j}{\sqrt{\Lambda_{jj}}} \stackrel{d}{\rightarrow} \mathcal{N}(0,1),$$
where $\Lambda_{jj}$ is the $j$-th diagonal entry of $\bm{\Lambda}$. Since $\bm{\Lambda}$ is unknown in practice, we estimate it by $\widehat{\bm\Lambda}$ and use the estimate to construct the test statistics.
For $j\in[p]$, define the test statistic as
$$U_j = \frac{\sqrt{n} \widehat{\theta}_j^d}{\sqrt{\widehat{\Lambda}_{jj}}},$$
where $\widehat{\Lambda}_{jj}$ is the $j$-th diagonal entry of $\widehat{\bm{\Lambda}}$.  Under the null hypothesis $\mathcal{H}_j^0:\theta_{0,j}=0$, Theorem \ref{theorem2} establishes that 
$U_j ={v_j}/{{\widehat{\Lambda}_{jj}}}^{1/2} + {\Delta_j}/{{\widehat{\Lambda}_{jj}}}^{1/2}\stackrel{d}{\rightarrow} \mathcal{N}(0,1),$
by the definition of $\widehat{\theta}_j^d$ under some mild conditions.

We propose first conducting hypothesis testing on all the high-dimensional main effects and then testing the interactions only if their corresponding high-dimensional main effects have been rejected. More specifically, with a proper threshold $t\geq0$, we reject the null hypotheses on the main effects if $|U_j|\geq t$ for $j\in[d]$ and obtain the rejected hypothesis index set $\mathcal{A}_1$. Then we test the null hypotheses on the interactions $\mathcal{H}_k^0: \theta_{0,k}= 0, k\in\mathcal{B}_{2j}$  if $j\in\mathcal{A}_1$ and reject
$\mathcal{H}_k^0$ if both $|U_k| \geq t$ and $|U_j| \geq t$. Denote the total number of rejections at threshold $t$ as 
$$R(t) = \sum_{j\in[d]} 1(|U_j|\geq t) + \sum_{j\in[d]} \sum_{k \in \mathcal{B}_{2j}} 1(|U_j|\geq t, |U_k|\geq t).$$
Then the false discovery proportions (FDPs) at threshold $t>0$ are given by:
$$\text{FDP}_0(t) = \frac{\sum_{j\in \mathcal{S}^c \cap [d]} 1(|U_j| \geq t)}{R(t)\vee 1},$$
$$\text{FDP}_j(t) = \frac{\sum_{k\in \mathcal{S}^c \cap \mathcal{B}_{2j}} 1(|U_j| \geq t, |U_k|\geq t)}{R(t)\vee 1},~ \text{for} ~ j \in[d].$$

With a pre-specified target $\alpha$, the threshold $t_0$ is determined by:
$$t_0 = \inf \left\{t \geq 0: \sum_{j=0}^d\text{FDP}_j(t) \leq \alpha \right\}.$$
With unknown $\mathcal{S}^c$, $\text{FDP}_j$'s are not directly available.  Following a common strategy in FDR control studies, we substitute $\sum_{j\in \mathcal{S}^c \cap [d]} 1(|U_j| \geq t)$ by $dG(t)$ and $\sum_{k\in \mathcal{S}^c \cap \mathcal{B}_{2j}} 1(|U_j| \geq t, |U_k|\geq t)$ by $qG^2(t)$ for conservative upper bounds, where $G(t) = 2(1-\Phi(t))$ and $\Phi(\cdot)$ is the cumulative density function of the standard normal distribution. 
Note that we assume $|\mathcal{S}^c| \rightarrow \infty$ as $p \rightarrow \infty$; otherwise, the FDR control problem reduces to a trivial strategy as we can just select all the covariates.
The detailed testing procedure is summarized in Algorithm \ref{Alg:1}. 
{\color{blue} The corresponding R code is available at \url{https://github.com/weijuanliang12138/Hierarchical-FDR-R-Code}.} 
In the next section, we establish the theoretical validity of the proposed procedure.

\begin{algorithm}
\caption{Hierarchical False Discovery Rate Control}
\label{Alg:1}
\begin{enumerate}
\item[Step 1:] Consider a pre-specified target level $\alpha \in[0,1]$. Let $t_p = (2\log p -2\log \log p)^{1/2}$, and calculate
\begin{equation}
\label{eq:algorithm1}
t_0 = \inf \left\{ 0  \leq t \leq t_p: \frac{dG(t)(1+qG(t)) }{R(t) \vee 1} \leq \alpha \right\}.
\end{equation}
If (\ref{eq:algorithm1}) does not exist, then set $t_0=\sqrt{2\log p}$.
\item[Step 2:] For the main effects, reject $\mathcal{H}_j^0$ if $|U_j| \geq t_0$ for $j\in[d]$, and obtain the index set of the rejected null hypothesis $\mathcal{A}_1$.
\item[Step 3:] For the interactions, test $\mathcal{H}_k^0$ if $k \in \mathcal{B}_{2j}, j\in \mathcal{A}_1$, and reject $\mathcal{H}_k^0$ if both $|U_j|\geq t_0$ and $|U_k| \geq t_0$.
\end{enumerate}
\end{algorithm}

\section{Asymptotic properties}
\label{sec:3}

We first establish the asymptotic normality of the debiased Lasso estimator $\widehat{\bm{\theta}}$ and then show that the proposed procedure achieves an asymptotic FDR control hierarchically under high-dimensional settings.

Let $H$ and $G$ be the distribution functions of $Y$ and $C$, respectively. Denote the end points of the support of $Y$, $T$, and $C$ as $\tau_Y$, $\tau_T$, and $\tau_C$, respectively. Denote the joint distribution of $(\Phi, T)$ as $F$. Following \cite{stute1996distributional}, we write:
$$
\widetilde{F}(\bm\phi, t)= \begin{cases}F(\bm\phi, t), & t<\tau_{Y}, \\ F(\bm\phi, {\tau_{Y^-}})+1({\tau_{Y}} \in \mathcal{D}) F(\bm\phi,\tau_{Y}), & t \geq \tau_{Y}, \end{cases}
$$
where $\mathcal{D}$ is the set of atoms of $H$. Define sub-distribution functions:
\begin{equation*}
\begin{aligned}
\widetilde{H}^{11}(\bm\phi, y)&=P(\Phi \leq \bm\phi, Y \leq y, \delta=1),\\
\widetilde{H}^{0}(y)&=P(Y \leq y, \delta=0).
\end{aligned}
\end{equation*}
For $j\in [p]$, let $\varphi_{j}(\Phi, Y)=\phi_{j}(Y-\Phi^\top \bm\theta_{0})$,
\begin{equation*}
\begin{aligned}
&\tau_{0}(y)=\exp \left(\int_{0}^{y-} \frac{\widetilde{H}^{0}(d s)}{1-H(s)}\right), \\
&\tau_{1}^{j}(y)=\frac{1}{1-H(y)} \int 1{(s>y)} \varphi_{j}(\bm\phi, s) \tau_{0}(s) \widetilde{H}^{11}(d \bm\phi, d s), \\
&\tau_{2}^{j}(y)=\iint \frac{1{(v < y, v<s)} \varphi_{j}(\bm\phi, s) \tau_{0}(s)}{(1-H(s))^{2}} \widetilde{H}^{0}(d v) \widetilde{H}^{11}(d \bm\phi, d s),
\end{aligned}
\end{equation*}
and $\zeta_{j}=\varphi_{j}(\Phi, Y) \tau_{0}(Y) \delta+\tau_{1}^{j}(Y)(1-\delta)-\tau_{2}^{j}(Y)$. Let $\sigma_{jk}=\text{cov}(\zeta_{j}, \zeta_{k})$ be the covariance of $\zeta_{j}$ and $\zeta_{k}$, and $\mathbf{\Sigma} = (\sigma_{jk})_{p\times p}$. The following conditions are assumed.

%condition 1
\begin{condition}
\label{cond:1}
 $P(T\leq C|T, X,Z) = P(T\leq C |T)$. 
\end{condition}

%condition 2
\begin{condition}
\label{cond:2}
The error term $\varepsilon$ has a sub-Gaussian distribution. There exist some constants $c_1, c_2>0$ such that $|X_j|\leq c_1$ and $|Z_k|\leq c_2$ for all $j\in[d]$ and $k\in[q]$.
\end{condition}

%condition 3
\begin{condition}
\label{cond:3}
Let $\bm{\Gamma}=\text{E}(\Phi \Phi^{\top})$ and assume that $\bm\Gamma$ satisfies the restricted eigenvalue condition:
$$
\kappa^{2}(\bm{\Gamma},\mathcal{S})=\inf_{\substack{\|\bm{a}_{\mathcal{S}^c} \|_{1} \leq 3\|\bm{a}_{\mathcal{S}}\|_{1} \\ \|\bm{a}\|_2 \neq 0}} \frac{\bm a^{\top} \bm\Gamma \bm a}{\|\bm{a}_\mathcal{S}\|_{2}^{2}} \geq c_3>0 .
$$
Furthermore, the eigenvalues of $\mathbf{\Sigma}$ are positive and bounded away from 0 and infinity.
\end{condition}

%condition 4
\begin{condition}
\label{cond:4}
$\int |{\varphi}_{j}(\bm{\phi}, s)| C^{{1}/{2}}(s) \widetilde{F}(d \bm\phi, d s)<\infty$ and $E[\varphi_{j}(\Phi, Y) \tau_{0}(Y) \delta]^{2}<\infty$ for every $j\in[p]$, where $C(s)=\int_{0}^{s- } {G(d y)}/{\{[1-H(y)][1-G(y)]\}}$.
\end{condition}

Condition \ref{cond:1} assumes that $\delta$ is conditionally independent of the covariates $\Phi$ given the failure time $Y$. Combining with the independence of $Y$ and $C$, the distribution of the observed $(\Phi, Y, \delta)$ is uniquely determined.  We refer to \cite{stute1996distributional} and followup studies for detailed discussions.
Condition \ref{cond:2} has been commonly assumed in high-dimensional studies and ensures that the debiased estimator has a desirable convergence rate. Similar conditions can be found in \cite{chai2019inference, li2022high}. 
Condition \ref{cond:3} is the well-known restricted eigenvalue condition for $\bm{\Gamma}$, which is usually assumed in high-dimensional model selection studies. Condition \ref{cond:4} is needed to establish the asymptotic normality of the proposed estimator, following a similar strategy as in \cite{stute1996distributional,chai2019inference}.

Recall that $\bm{\Lambda} = \mathbf{M} \mathbf{\Sigma} \mathbf{M}^{\top}$, where $\mathbf{M} = \bm\Gamma^{-1}$. With the construction of $\widehat{\mathbf{M}}$ in (\ref{eq:2.5}), in what follows, Theorem \ref{theorem1} (a) establishes the entry-wise asymptotic normality of $\bm{v}$, and Theorem \ref{theorem1} (b) establishes that the ``noise" term $\bm{\Delta}$ is asymptotic negligible compared to $\bm{v}$. The proofs are provided in Appendix.

%there exist some constants $c_{\min}$, $c_{\max}$ such that
%$0 < c_{\min} \leq \sigma_{\min} (\mathbf{\Sigma}) \leq \sigma_{\max} (\mathbf{\Sigma}) \leq c_{\max} < \infty$,

% Theorem 1
\begin{theorem}
\label{theorem1}
Assume that Conditions \ref{cond:1}-\ref{cond:4} hold. If $|\mathcal{S}^c| \geq cp$ for a constant $0<c\leq 1$, $|\mathcal S| \sqrt{\log p / n}=o(1)$, $\lambda=O(\sqrt{\log p / n})$, and $\mu=O(\sqrt{\log {p} / n})$, then: 
 \begin{itemize}
 \item[(a)] $v_j \stackrel{d}{\rightarrow}  \mathcal{N}(0, \Lambda_{jj})$, where $ v_j=  \sqrt{n} \sum_{i=1}^n w_i \widehat{\bm{m}}_j^\top \bm{\phi}_i  \varepsilon_i$.
 \item[(b)] If additionally $\sqrt{n} \lambda \mu |\mathcal{S}| \rightarrow 0$, then  $\| \bm{\Delta} \|_\infty = o_p(1)$, where $\bm\Delta=\sqrt{n}(\widehat{\mathbf{M}} \widehat{\bm\Gamma}-\mathbf I)(\widehat{\bm\theta}-\bm\theta_{0})$.
 \end{itemize}
\end{theorem}

Define the normalized matrix $\bm{\Lambda}^0$ with
$\Lambda_{jk}^0 = {\Lambda_{jk}}/{\sqrt{\Lambda_{jj} \Lambda_{kk}}}$. For a given constant $\xi >0$, define $\Xi(\xi, b) \equiv \{ (i,j): 1\leq i,j \leq p, |\Lambda_{ij}^0| \geq b(\log p)^{-2-\xi} \}$ for some constant $b>0$. Theorem \ref{theorem2} below establishes that the proposed procedure controls the hierarchical FDR at level $\alpha$.

\begin{theorem}
\label{theorem2}
Assume that the conditions of Theorem \ref{theorem1} hold. If there exist some positive constants $b, \xi$, such that $|\Xi(\xi, b)| = o(p^{1+\rho})$ for some $\rho \in[0,1)$, and $|\{(i,j): |\Lambda_{ij}^0|> (1-\rho)/(1+\rho) \}| = O(p)$.
Then for the hierarchical FDR control procedure described in Algorithm \ref{Alg:1}, 
$$
\limsup \limits_{(n,d)\rightarrow \infty} \text{FDR} \leq \alpha.
$$
\end{theorem}

The magnitude of the cardinality $|\Xi(\xi, b)|$ measures the number of highly correlated test statistics. The condition $|\Xi(\xi, b)| = o(p^{1+\rho})$ is mild since $b(\log p)^{-2-\xi}$ goes to zero at a slow logarithmic rate and the cardinality $|\Xi(\xi, b)|$ is smaller than a polynomial order $p^{1+\rho}$ for some $\rho \in[0,1)$. Theorem \ref{theorem2} guarantees that the proposed procedure controls FDR hierarchically at the pre-assigned level $\alpha$. %In the next section, we show that this is empirically true and the proposed procedure has a relative high power.

\section{Simulation}
\label{sec:4}

Data is generated from model (\ref{eq:2.1}) where the high-dimensional main effect design matrix $\mathbf{X}\in \mathbb{R}^{n\times d}$ is generated by drawing its rows independently from $\mathcal{N}(0,\bm{\Sigma}_{X})$. The covariance matrix $\bm{\Sigma}_X$ has an auto-correlation structure with the $(i,j)$-th entry being $\eta^{|i-j|}$ for some constant $\eta \in(0,1)$.
The low-dimensional main effect design matrix $\mathbf{Z} \in \mathbb{R}^{n\times q}$ is generated from $q=5$  independent standard normal distributions.  The augmented design matrix is $\bm{\Phi}\in\mathbb{R}^{n\times p}$. The values of the nonzero entries of the true coefficient vector $\bm{\theta}_0$ are 2 for the main effects and 1 for the interactions. The nonzero coefficient index set $\mathcal{S}$ varies case-by-case.  
The failure times are exponentially distributed with rates $\exp(-a\bm{\Phi}^\top \bm{\theta}_0)$. The magnitude of $a$ controls the strength of the signals.  The censoring times are generated independently and adjusted to achieve certain censoring rates.

We conduct multiple sets of simulation to compare the performance of the proposed procedure with six highly relevant alternatives and examine the impact of sample size ($n$), censoring rate ($r$), sparsity level ($s$), dimension ($p$), correlation ($\eta$), and signal magnitude ($a$). Specifically, the alternatives include: BH and BH-Hierarchy, { which apply the Benjamini-Hochberg (BH) technique and derive p-values from the conventional univariable AFT model and the marginal weighted least squared estimation without penalization \citep{benjamini1995controlling}. The R function \emph{p.adjust} in the package \textbf{stats} is employed.} The latter conducts hierarchical FDR controlling; Surv-FCD, which is similar to the proposed procedure except that it controls the non-hierarchical FDR. 
{ The debiased Lasso estimator is obtained using the code from 
\url{https://web.stanford.edu/~montanar/sslasso/code.html};} 
{ The rest three alternatives conduct variable selection (VS) without an explicit FDR control. With their popularity, comparing with them can help better benchmark the proposed method.
} 
VS-D-Lasso, VS-Lasso, and VS-MCP conduct penalized estimation and selection using the debiased Lasso, Lasso, and MCP techniques, respectively. 
{
The VS-MCP method replaces the Lasso penalty in (\ref{eq:2.3}) with the MCP 
$\rho(\theta_j; \lambda, \xi)=\lambda  \int_0^{|\theta_j|} (1-x/(\lambda \xi))_+dx$,  where $\xi$ is a regularization parameter. We refer to \cite{chai2019inference, zhang2010nearly} for more information on the computation and other aspects of MCP. The R function \emph{ncvsurv} in the package \textbf{ncvreg} is utilized to realize the VS-Lasso and VS-MCP methods.}
We set the target FDR as $\alpha = 0.1$ and examine empirical FDR and power based on {{200}} replicates under each setting. 

Denote the number of nonzero coefficients corresponding to $\mathbf{X}$ as $s_\alpha$. The nonzero indices of $\bm{\theta}_0$ are generated as $\mathcal{S} =  [s_\alpha]\cup \{d+2, d+5\}\cup \{( d+kq +2, d+kq+5): k\in [s_\alpha] \}$. $\mathcal{S}$ corresponds to the first  $s_\alpha$ entries of $\{\alpha_{0,j}\}_1^d$ as well as the second and fifth entries of $\{\gamma_{0,j}\}_1^q$, and the interactions associated with them are nonzero. Unless otherwise specified, we set $n = 500$, $r = 0.2$, $s_\alpha = 10$, $p = 1205$, $\eta = 0.3$,  $a = 1$, and $d = 200$.  Note that $s = 3s_\alpha+2$ and $p=6d+5$. The results for $a=1$ are presented in Figures \ref{fig1} and \ref{fig2} in the main text, and those for $a=2$ are presented in Figures \ref{fig3} and \ref{fig4} in Appendix. Among them, Figures \ref{fig1} and \ref{fig3} are on empirical FDR, and Figures \ref{fig2} and \ref{fig4} are on empirical power. 

\textbf{Effect of sample size:} We set $n = 300$, 500, 700, 900, and $1100$. As Figure \ref{fig2} shows, the power of all methods grows as sample size increases, with the proposed procedure having relatively competitive power while controlling FDR hierarchically even when the sample size is small. Note that when $n = 300$, only the proposed procedure controls FDR below the target, while all the others fail (see Figure \ref{fig1}). 

\textbf{Effect of censoring rate:} We consider $r = 0$, 0.1, $\cdots$, 0.7. As expected, performance deteriorates as censoring increases, {which may be attributable to the decrease in effective sample size}. When $r=0.7$, only the proposed procedure can successfully control FDR. It is also observed to have satisfactory power.

\textbf{Effect of sparsity level:} We set $s_\alpha=$10, 20, 30, 40, and 50, which correspond to $s=32$, 62, 92, 122, and 152, respectively. As Figure \ref{fig2} shows, power decreases as the number of nonzero effects increases, while FDR control remains stable. {The estimation of $\bm{\theta}$ deteriorates as $s$ increases, which fits Theorem \ref{theorem1} and previous studies \citep{chai2019inference, jankova2016confidence}.} Surv-FCD, VS-MCP, and the proposed procedure control FDR for all $s$ values. The proposed procedure has the highest power. 

\textbf{Effect of dimension:} We set $d = 100$, 200, 300, 400, and 500, which correspond to $p = 605$, 1205, 1805, 2405, and 3005, respectively. { It has been well recognized that a higher dimension usually poses more challenges to estimation and FDR control.} When the dimension is high, some of the methods (BH, BH-Hierarchy, and VS-D-Lasso) fail to control FDR. The proposed procedure, Surv-FCD, VS-MCP, and VS-Lasso control FDR for all dimension values, and all of these methods have power approaching 1.

\textbf{Effect of correlation:} We consider $\eta$ in the set $\{0.1, 0.2,\cdots,0.8\}$. { As correlation increases, it is more difficult to distinguish between important effects and noises.} It is observed that the proposed procedure and VS-Lasso control FDR for all correlation levels while maintaining power at a desirable level. The other methods fail to control FDR.

\textbf{Effect of signal:}  We consider $a= 0.5$, 1, $\cdots$, 3. { Stronger signals are easier to identify, which may consequently lead to higher empirical power under a pre-specified FDR level.} It is observed that empirical FDR becomes less conservative with the increase of signal strength except for the BH and BH-Hierarchy methods, since these two methods conduct marginal analysis. Again, the proposed procedure is observed to have competitive performance.

{
To gain additional insight, we also examine the performance of the proposed 
method and its alternatives under different pre-specified FDR levels and a different survival data distribution. In particular, we set the FDR levels as 0.05, 0.1, 0.15, 0.20, 0.25, and 0.30. 
We consider two survival models: under the first model, the survival time is exponentially distributed with rate $\exp(-\bm{\Phi}^\top \bm{\theta}_0)$; and under the second model, it is log-logistic distributed with hazard function $1/(\exp{(\bm{\Phi}^\top \bm{\theta}_0)}+T)$. Note that the first model satisfies both the Proportional Hazards (PH) and Accelerated Failure Time (AFT) assumptions, while the latter only satisfies the AFT assumption. The results for the empirical FDR and power are presented in Figure \ref{fig5} in Appendix. The average mean square errors (MSEs) of all the parameters and runtimes are presented in Table \ref{sim:1}. With the exponentially distributed failure time, the proposed method, Surv-FCD, VS-Lasso, and VS-MCP control FDR satisfactorily for all the pre-specified FDR values, with the empirical power approaching 1. As shown in Table \ref{sim:1}, incorporating the debiased estimation results in better parameter estimation performance in terms of MSEs, and the performance may be further improved by introducing the control of FDR. It is noted that VS-MCP has superior performance, which aligns with observations in the literature. However, it can be very challenging to establish the asymptotic distribution properties of MCP estimates. 
Under the log-logistic model, only the proposed and VS-Lasso methods can control FDR successfully for all the pre-specified FDR levels, and the proposed method significantly outperforms the VS-Lasso method in terms of MSEs. Table \ref{sim:1} shows that the marginal methods (BH and BH-Hierarchy) are the fastest, and the proposed method has a computational cost similar to that of Surv-FCD and VS-D-Lasso. And these three methods are only slightly slower than VS-MCP.
}

\section{Analysis of breast cancer data}
\label{sec:5}

To demonstrate the practical utility of the proposed approach, we analyze a breast cancer dataset. There has been extensive research linking breast cancer survival with omics measurements, and genetic interactions have been established as playing an important role beyond the main genetic effects. The analyzed data is retrieved from the Molecular Taxonomy of Breast Cancer International Consortium (METABRIC) \url{https://www.cbioportal.org/study/summary?id=brca\_metabric} and has been analyzed in the literature \citep{curtis2012genomic, rueda2019dynamics}. The dataset contains records on 1,903 subjects. The outcome of interest is overall survival, which is subject to right censoring and has been extensively studied. 1,103 patients died during followup, and the event times {range} from 0.1 to 355.2 months, with a median of 85.93 months. The rest 800 observations are censored, and the observed times range from 0.77 to 337.03 months, with a median of 158.03 months. 

For the high-dimensional covariates, we consider gene expressions measured using the RNA-seq technique. Measurements are available for 18,485 genes. Considering the limited sample size, and to generate more reliable results, we first conduct screening as follows. We retrieve gene pathway information from the Kyoto Encyclopedia of Genes and Genomes (KEGG) and focus on the following pathways that are more likely to be breast cancer relevant: breast cancer pathway (hsa05224), homologous recombination pathway (hsa03440), fanconi anemia pathway (hsa03460), PI3K-Akt signaling pathway (hsa04151), and pathways in cancer (hsa05200). A total of 782 genes are matched to those five pathways. Additionally, we conduct a supervised screening and retain the genes with marginal correlations greater than 0.05 with the outcome, which leads to 312 genes for downstream analysis. It is noted that similar screenings have been common in the literature.
For the low-dimensional covariates, we consider nonsynonymous Tumor Mutation Burden (TMB, which is defined as the number of somatic nonsynonymous mutations divided by the DNA sequenced megabase and ranges from 0 to 104.601 mutations/Mb), age at diagnosis (ranging from 21.93 to 96.29 years), and Estrogen Receptor status (ER, an indicator of endocrine responsiveness, and with 1 for positive and 0 for negative). The total number of effects, main and interaction combined, is 1,251.

The analysis results using the proposed procedure { with a  target FDR level $\alpha =0.2$} are presented in Table  \ref{realdata:1}. A total of 10 genes and 3 interactions are identified. A {quick} literature search suggests that these genes have important implications for breast cancer survival {(details are presented in Appendix)}. 
Analysis is also conducted using the alternatives, and the findings are presented in the Appendix. The summary comparison results are presented in Table \ref{realdata:2}, where we observe that different methods lead to quite different findings. The overall observed patterns are consistent with the simulation study.

\begin{figure}[h]
\centering
\includegraphics[height=3.8 in, width=5.0 in]{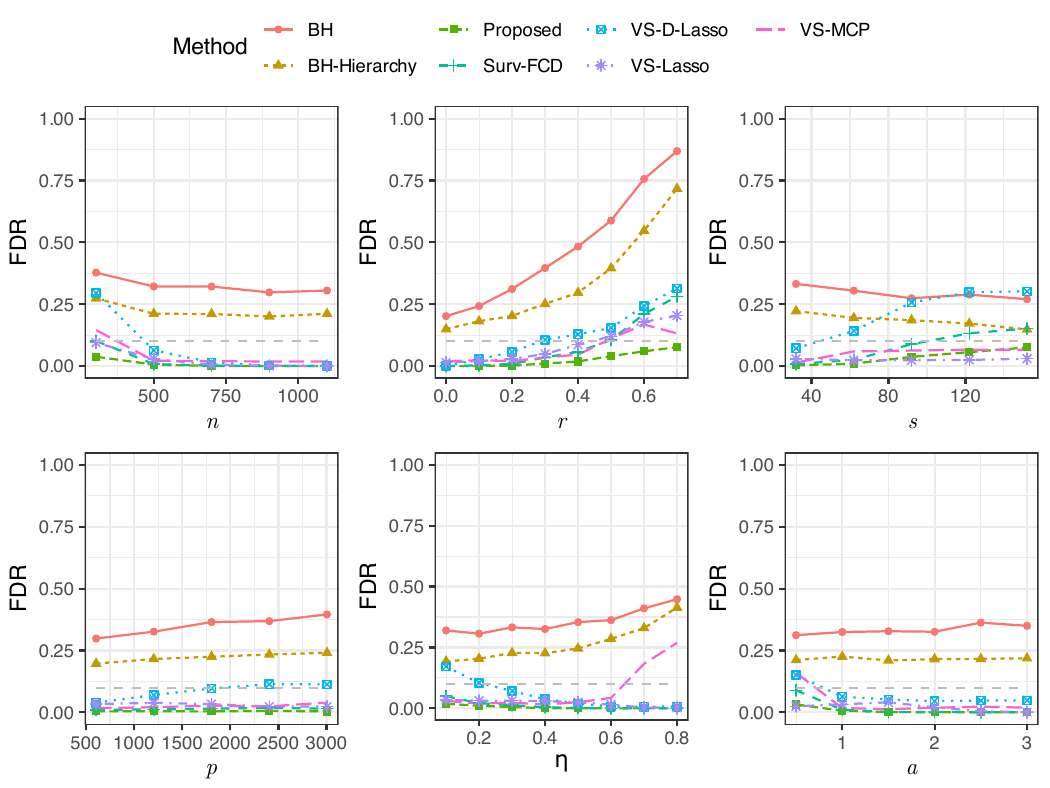}
\caption{Empirical FDR based on 200 replicates with $\alpha=0.1$ and $a=1$ (except for the right bottom subfigure with different $a$'s). 
Top panels (from left to right) correspond to various sample sizes, censoring rates, and numbers of nonzero covariates, respectively. Bottom panels (from left to right) correspond to different dimensions, correlation coefficients, and signal magnitudes, respectively. Dashed grey lines correspond to target FDR.}
\label{fig1}
\end{figure}

\begin{figure}[h]
\centering
\includegraphics[height=3.8 in, width=5.0 in]{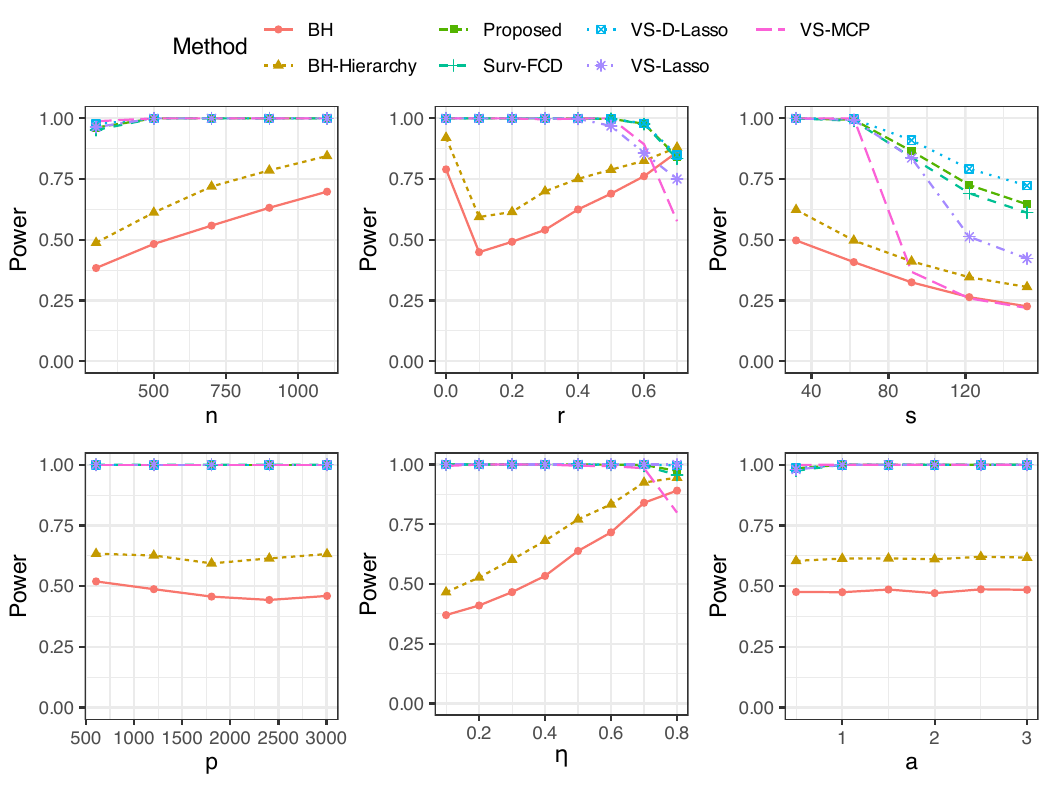}
\caption{Empirical power based on 200 replicates with $\alpha=0.1$ and $a=1$ (except for the right bottom subfigure with different $a$'s). 
%Dashed grey lines: target FDR; Solid red circles: BH; Dashed orange triangles: BH-Hierarchy; Dashed green squares: proposed; Solid cyan crosses: Surv-FCD; Dashed wathet cross squares: VS-D-Lasso; Dashed purple stars: VS-Lasso; Dashed magenta: VS-MCP. 
Top panels (from left to right) correspond to various sample sizes, censoring rates, and numbers of nonzero covariates, respectively. Bottom panels (from left to right) correspond to different dimensions, correlation coefficients, and signal magnitudes, respectively.
}
\label{fig2}
\end{figure}

\clearpage
\begin{table}
\setlength{\abovecaptionskip}{0pt}
\setlength{\belowcaptionskip}{5pt}
\caption{Simulation: average MSEs and runtimes based on 200 replicates under different pre-specified FDR level. All methods are implemented with R 4.3.1 on an MacOS computer with an 
Apple M1 Max chip, 10-core CPU, 32 GB of RAM, and 1 TB of disk space.}
\label{sim:1}
\centering
\resizebox{320pt}{55mm}{
\begin{tabular}{cccccccc}
\hline
Survival model                 & Method       & \multicolumn{6}{c}{Pre-specified FDR level $\alpha$}                        \\
\cline{3-8}
                               &              & 0.05   & 0.1    & 0.15   & 0.2    & 0.25   & 0.3    \\
                               \hline
\multirow{14}{*}{Exponential}  &              & \multicolumn{6}{c}{MSEs($\times 10^{3}$)}            \\
                               & Proposed     & 2.477  & 2.478  & 2.479  & 2.479  & 2.479  & 2.479  \\
                               & Surv-FCD     & 2.482  & 2.484  & 2.486  & 2.487  & 2.489  & 2.490  \\
                               & VS-D-Lasso   & 2.513  & 2.513  & 2.513  & 2.513  & 2.513  & 2.513  \\
                               & VS-Lasso     & 14.048 & 14.118 & 14.110 & 14.202 & 13.852 & 14.148 \\
                               & VS-MCP       & 0.351  & 0.357  & 0.345  & 0.349  & 0.349  & 0.341  \\
                               &              & \multicolumn{6}{c}{Runtime (seconds)}              \\
                               & BH           & 0.974  & 0.963  & 0.963  & 0.963  & 0.964  & 0.965  \\
                               & BH-Hierarchy & 0.316  & 0.316  & 0.320  & 0.324  & 0.331  & 0.336  \\
                               & Proposed     & 26.927 & 26.493 & 25.824 & 25.057 & 25.042 & 24.613 \\
                               & Surv-FCD     & 26.863 & 26.428 & 25.759 & 24.993 & 24.979 & 24.546 \\
                               & VS-D-Lasso   & 26.824 & 26.389 & 25.721 & 24.954 & 24.940 & 24.508 \\
                               & VS-Lasso     & 7.219  & 7.227  & 7.237  & 7.250  & 7.244  & 7.253  \\
                               & VS-MCP       & 16.922 & 16.921 & 16.930 & 16.978 & 16.951 & 16.925 \\
                               \hline
\multirow{14}{*}{Log-logistic} &              & \multicolumn{6}{c}{MSEs($\times 10^{3}$)}            \\
                               & Proposed     & 4.203  & 4.015  & 4.015  & 3.915  & 3.930  & 3.852  \\
                               & Surv-FCD     & 4.271  & 4.089  & 4.090  & 3.997  & 4.014  & 3.927  \\
                               & VS-D-Lasso   & 4.363  & 4.162  & 4.161  & 4.056  & 4.072  & 3.979  \\
                               & VS-Lasso     & 24.803 & 24.014 & 24.138 & 24.118 & 23.883 & 24.424 \\
                               & VS-MCP       & 8.512  & 8.200  & 8.502  & 8.130  & 8.188  & 8.899  \\
                               &              & \multicolumn{6}{c}{Runtime  (seconds)}              \\
                               & BH           & 0.967  & 0.922  & 0.917  & 0.918  & 0.922  & 0.916  \\
                               & BH-Hierarchy & 0.319  & 0.304  & 0.308  & 0.311  & 0.320  & 0.322  \\
                               & Proposed     & 25.004 & 24.286 & 24.031 & 24.151 & 24.087 & 23.641 \\
                               & Surv-FCD     & 24.911 & 24.205 & 23.949 & 24.070 & 24.008 & 23.561 \\
                               & VS-D-Lasso   & 24.867 & 24.169 & 23.912 & 24.033 & 23.971 & 23.525 \\
                               & VS-Lasso     & 7.261  & 7.166  & 7.306  & 7.319  & 7.388  & 7.411  \\
                               & VS-MCP       & 19.678 & 19.568 & 19.845 & 19.650 & 19.436 & 19.674 	\\
\hline                               
\end{tabular}
}
\end{table}

\clearpage
\begin{table}[]
\setlength{\abovecaptionskip}{0pt}
\setlength{\belowcaptionskip}{5pt}
\caption{Analysis of the breast cancer data using the proposed approach: identified main effects and interactions.}
\label{realdata:1}
\centering
\resizebox{210pt}{25mm}
{
\begin{tabular}{ccccc}
\hline
Gene            & Main effects & \multicolumn{3}{c}{ Interactions} \\
\cline{3-5}
                &           & TMB           & Age           & ER Status        \\
                \hline
COL6A6          & 0.113     &               &               &                  \\
IBSP            & -0.075    &               & 0.048         &                  \\
CDK2            & -0.053    &               &               &                  \\
CCND2           & 0.089     &               &               &                  \\
NR4A1           & -0.043    &               &               &                  \\
WNT6            & -0.030    & -0.039        &               &                  \\
CDC42           & -0.046    &               &               &                  \\
ESR1            & -2.322    &               &               & 2.442            \\
GSTM2           & 0.042     &               &               &                  \\
DLL3            & -0.038    &               &               &                  \\
\hline
                &           & \multicolumn{3}{c}{Low-dimensional covariates}                  \\
                \cline{3-5}
                &           & TMB           & Age           & ER Status        \\
\hline
Main effects &           & 0.045         & -0.064        & 4.701       \\
\hline    
\end{tabular}
}
\end{table}

\begin{table}[]
\setlength{\abovecaptionskip}{0pt}
\setlength{\belowcaptionskip}{5pt}
\caption{Analysis of the breast cancer data: numbers of main G effects and interactions identified by different methods and their overlaps. In each cell, number of identified main G effects/(number of identified interactions, number of identified interactions that respect the ``main effect, interactions" hierarchy).}
\label{realdata:2}
\centering
\resizebox{270pt}{11mm}
{
\begin{tabular}{cccccc}
\hline
         & Proposed & Surv-FCD & VS-D-Lasso   & VS-Lasso     & VS-MCP     \\
         \hline
Proposed & 10/(3,3) & 3/(2,1)  & 10/(3,3)  &  10/(3,3) & 3/(1,1) \\
Surv-FCD &          & 3/(7,1)  & 3/(7,1)   &  3/(7,1) & 2/(1,1)  \\
VS-D-Lasso  &          &          & 31/(30,3) & 22/(21,3) & 3/(1,1) \\
VS-Lasso       &          &          &           & 24/(21,3) & 3/(1,1)   \\
VS-MCP   &          &          &           &         & 4/(2,1)\\
\hline
\end{tabular}
}
\end{table}

\section{Discussion}
In this article, we have developed a new hierarchical FDR control approach for the analysis of a survival response and high-dimensional interactions. The proposed approach has a solid theoretical ground and satisfactory empirical performance. This study can enrich the paradigms of high-dimensional interaction analysis and FDR-based inference. 
{ 
In principle, the proposed analysis procedure can be directly extended to other data types/models as well as other types of interactions (for example, gene-gene interactions). However, as some theoretical developments are specific to the proposed model, significant nontrivial developments will be needed. Additionally, more numerical studies may be needed to better understand the finite-sample performance of the proposed approach.
}

\section*{Acknowledgments}
We thank the editor and reviewers for their careful review and insightful comments. 
This study is partly supported by the National Bureau of Statistics of China (2022LZ34), NIH (CA204120), Fundamental Research Funds for the Central Universities,  and Research Funds of the Renmin University of China (21XNH152).

\section*{Reference}
\bibliography{Citation}

\begin{thebibliography}{10}
\expandafter\ifx\csname url\endcsname\relax
  \def\url#1{\texttt{#1}}\fi
\expandafter\ifx\csname urlprefix\endcsname\relax\def\urlprefix{URL }\fi
\expandafter\ifx\csname href\endcsname\relax
  \def\href#1#2{#2} \def\path#1{#1}\fi

\bibitem{he2022rank}
B.~He, S.~Ma, X.~Zhang, L.~Zhu, Rank-based greedy model averaging for
  high-dimensional survival data, Journal of the American Statistical
  Association (2022) 1--13.

\bibitem{chai2019inference}
H.~Chai, Q.~Zhang, J.~Huang, S.~Ma, Inference for low-dimensional covariates in
  a high-dimensional accelerated failure time model, Statistica Sinica 29~(2)
  (2019) 877--894.

\bibitem{vansteelandt2022assumption}
S.~Vansteelandt, O.~Dukes, K.~Van~Lancker, T.~Martinussen, Assumption-lean cox
  regression, Journal of the American Statistical Association (2022) 1--10.

\bibitem{wu2020structured}
M.~Wu, Q.~Zhang, S.~Ma, Structured gene-environment interaction analysis,
  Biometrics 76~(1) (2020) 23--35.

\bibitem{ren2023robust}
J.~Ren, F.~Zhou, X.~Li, S.~Ma, Y.~Jiang, C.~Wu, Robust bayesian variable
  selection for gene-environment interactions, Biometrics 79~(2) (2023)
  684--694.

\bibitem{haris2016convex}
A.~Haris, D.~Witten, N.~Simon, Convex modeling of interactions with strong
  heredity, Journal of Computational and Graphical Statistics 25~(4) (2016)
  981--1004.

\bibitem{du2021lasso}
Y.~Du, H.~Chen, R.~Varadhan, Lasso estimation of hierarchical interactions for
  analyzing heterogeneity of treatment effect, Statistics in Medicine 40~(25)
  (2021) 5417--5433.

\bibitem{wang2022two}
J.~Wang, A.~Patel, J.~M. Wason, P.~J. Newcombe, Two-stage penalized regression
  screening to detect biomarker-treatment interactions in randomized clinical
  trials, Biometrics 78~(1) (2022) 141--150.

\bibitem{wu2018dissecting}
C.~Wu, Y.~Jiang, J.~Ren, Y.~Cui, S.~Ma, Dissecting gene-environment
  interactions: A penalized robust approach accounting for hierarchical
  structures, Statistics in Medicine 37~(3) (2018) 437--456.

\bibitem{xu2022multidimensional}
Y.~Xu, M.~Wu, S.~Ma, Multidimensional molecular measurements--environment
  interaction analysis for disease outcomes, Biometrics 78~(4) (2022)
  1542--1554.

\bibitem{bien2013lasso}
J.~Bien, J.~Taylor, R.~Tibshirani, A lasso for hierarchical interactions, The
  Annals of Statistics 41~(3) (2013) 1111--1141.

\bibitem{javanmard2014confidence}
A.~Javanmard, A.~Montanari, Confidence intervals and hypothesis testing for
  high-dimensional regression, The Journal of Machine Learning Research 15~(1)
  (2014) 2869--2909.

\bibitem{dai2022false}
C.~Dai, B.~Lin, X.~Xing, J.~S. Liu, False discovery rate control via data
  splitting, Journal of the American Statistical Association (2022) 1--38.

\bibitem{candes2018panning}
E.~Candes, Y.~Fan, L.~Janson, J.~Lv, Panning for gold: ‘model-x’ knockoffs
  for high dimensional controlled variable selection, Journal of the Royal
  Statistical Society: Series B (Statistical Methodology) 80~(3) (2018)
  551--577.

\bibitem{yekutieli2008hierarchical}
D.~Yekutieli, Hierarchical false discovery rate--controlling methodology,
  Journal of the American Statistical Association 103~(481) (2008) 309--316.

\bibitem{lynch2017control}
G.~Lynch, W.~Guo, S.~K. Sarkar, H.~Finner, The control of the false discovery
  rate in fixed sequence multiple testing, Electronic Journal of Statistics
  11~(2) (2017) 4649--4673.

\bibitem{bogomolov2021hypotheses}
M.~Bogomolov, C.~B. Peterson, Y.~Benjamini, C.~Sabatti, Hypotheses on a tree:
  new error rates and testing strategies, Biometrika 108~(3) (2021) 575--590.

\bibitem{g2016sequential}
M.~G. G'Sell, S.~Wager, A.~Chouldechova, R.~Tibshirani, Sequential selection
  procedures and false discovery rate control, Journal of the Royal Statistical
  Society: Series B (Statistical Methodology) 78~(2) (2016) 423--444.

\bibitem{wei1992accelerated}
L.~J. Wei, The accelerated failure time model: a useful alternative to the cox
  regression model in survival analysis, Statistics in Medicine 11~(14-15)
  (1992) 1871--1879.

\bibitem{van2014asymptotically}
S.~Van~de Geer, P.~B{\"u}hlmann, Y.~Ritov, R.~Dezeure, On asymptotically
  optimal confidence regions and tests for high-dimensional models, The Annals
  of Statistics 42~(3) (2014) 1166--1202.

\bibitem{javanmard2019false}
A.~Javanmard, H.~Javadi, False discovery rate control via debiased lasso,
  Electronic Journal of Statistics 13~(1) (2019) 1212--1253.

\bibitem{stute1993consistent}
W.~Stute, Consistent estimation under random censorship when covariables are
  present, Journal of Multivariate Analysis 45~(1) (1993) 89--103.

\bibitem{stute1996distributional}
W.~Stute, Distributional convergence under random censorship when covariables
  are present, Scandinavian Journal of Statistics (1996) 461--471.

\bibitem{li2022high}
D.~Li, Y.~Kong, Y.~Fan, J.~Lv, High-dimensional interaction detection with
  false sign rate control, Journal of Business and Economic Statistics 40~(3)
  (2022) 1234--1245.

\bibitem{benjamini1995controlling}
Y.~Benjamini, Y.~Hochberg, Controlling the false discovery rate: a practical
  and powerful approach to multiple testing, Journal of the Royal Statistical
  Society: Series B (Statistical Methodology) 57~(1) (1995) 289--300.

\bibitem{zhang2010nearly}
C.~Zhang, Nearly unbiased variable selection under minimax concave penalty, The
  Annals of Statistics 38~(2) (2010) 894--942.

\bibitem{jankova2016confidence}
J.~Jankov{\'a}, S.~Van De~Geer, Confidence regions for high-dimensional
  generalized linear models under sparsity, arXiv preprint arXiv:1610.01353.

\bibitem{curtis2012genomic}
C.~Curtis, S.~P. Shah, S.-F. Chin, G.~Turashvili, O.~M. Rueda, M.~J. Dunning,
  D.~Speed, A.~G. Lynch, S.~Samarajiwa, Y.~Yuan, et~al., The genomic and
  transcriptomic architecture of 2,000 breast tumours reveals novel subgroups,
  Nature 486~(7403) (2012) 346--352.

\bibitem{rueda2019dynamics}
O.~M. Rueda, S.~J. Sammut, J.~A. Seoane, S.~F. Chin, J.~L. Caswell~Jin,
  M.~Callari, R.~Batra, B.~Pereira, A.~Bruna, H.~R. Ali, et~al., Dynamics of
  breast-cancer relapse reveal late-recurring er-positive genomic subgroups,
  Nature 567~(7748) (2019) 399--404.

\bibitem{dai2023scale}
C.~Dai, B.~Lin, X.~Xing, J.~S. Liu, A scale-free approach for false discovery
  rate control in generalized linear models, Journal of the American
  Statistical Association 118~(543) (2023) 1551--1565.

\bibitem{liu2013gaussian}
W.~Liu, Gaussian graphical model estimation with false discovery rate control,
  The Annals of Statistics 41~(6) (2013) 2948--2978.

\bibitem{berti2002uniform}
P.~Berti, P.~Rigo, A uniform limit theorem for predictive distributions,
  Statistics and Probability Letters 56~(2) (2002) 113--120.

\bibitem{yeh2018extracellular}
M.~H. Yeh, Y.~J. Tzeng, T.~Y. Fu, J.~J. You, H.~T. Chang, L.~P. Ger, K.~W.
  Tsai, Extracellular matrix--receptor interaction signaling genes associated
  with inferior breast cancer survival, Anticancer research 38~(8) (2018)
  4593--4605.

\bibitem{klein2014handbook}
J.~P. Klein, H.~C. Van~Houwelingen, J.~G. Ibrahim, T.~H. Scheike, Handbook of
  survival analysis, Boca Raton, FL: Chapman and Hall/CRC Press, 2014.

\bibitem{wu2021exosomal}
K.~Wu, J.~Feng, F.~Lyu, F.~Xing, S.~Sharma, Y.~Liu, S.~Y. Wu, D.~Zhao,
  A.~Tyagi, R.~P. Deshpande, et~al., Exosomal mir-19a and ibsp cooperate to
  induce osteolytic bone metastasis of estrogen receptor-positive breast
  cancer, Nature Communications 12~(5196) (2021) 1--18.

\bibitem{dustin2019esr1}
D.~Dustin, G.~Gu, S.~A. Fuqua, Esr1 mutations in breast cancer, Cancer 125~(21)
  (2019) 3714--3728.

\end{thebibliography}

\section*{Appendix}
\label{appendix}

\subsection*{Appendix A: Proofs of Theorems \ref{theorem1} and \ref{theorem2}}
\label{appendixA}

The compatibility condition for the sample weighted covariance matrix $\widehat{\bm{\Gamma}}$ is critical for the estimation error. To prove Theorem \ref{theorem1}, we first show that $\widehat{\bm{\Gamma}}$ satisfies the compatibility condition. With the high dimensionality, $\widehat{\bm{\Gamma}}$ is always singular, which makes the estimation of $\bm{\theta}_0$ challenging. A common assumption to deal with this problem is to require $\widehat{\bm{\Gamma}}$ to be nonsingular for a restricted set of directions. Under some mild conditions, we have the following lemma.

\begin{lemma}
\label{lemma1}
Under Conditions \ref{cond:1}-\ref{cond:3} and $|\mathcal{S}| \sqrt{\log p/n} = o(1)$, if $\max_{i,j\in[p]} |\widehat{\Gamma}_{ij} - {\Gamma}_{ij} |=O_p(\sqrt{\log p/n})$, we have %$\widehat{\bm{\Gamma}}$ satisfies
\begin{equation*}
\inf_{\substack{\|\bm{a}_{\mathcal{S}^c} \|_{1} \leq  3\| \bm{a}_{\mathcal{S}}\|_{1} \\ \| \bm{a} \|_2 \neq 0}}
 \frac{\bm a^{\top}\widehat{\bm{\Gamma}} \bm a}{\|\bm{a}_{\mathcal{S}} \|_{2}^{2}}\geq \frac{c_{3}}{2}>0, \quad \text { as } n \rightarrow \infty .
\end{equation*}
\end{lemma}

\begin{proof}[Proof of Lemma 1]
With $\| \bm{a}_{\mathcal{S}^c} \|_{1} \leq  3\| \bm{a}_\mathcal{S}\|_{1}$, we have:
\begin{equation*}
\begin{aligned}
\bm a^{\top} \bm\Gamma \bm a-\bm a^{\top} \widehat{\bm\Gamma} \bm a 
& \leq |\bm{a}^{\top}(\bm\Gamma-\widehat{\bm\Gamma}) \bm{a} | 
\leq \|\bm{a}\|_{1} \|(\bm\Gamma-\widehat{\bm\Gamma}) \bm{a}\|_{\infty} 
\leq \|\bm{a}\|_{1}^{2}  \max _{i,j \in [p]}|\widehat{\Gamma}_{ ij}-\Gamma_{i j}|\\
& \leq 16 \| \bm{a}_{\mathcal{S}} \|_{1}^2  \max _{i,j\in [p]}|\widehat{\Gamma}_{ ij}-\Gamma_{i j}|  
 \leq 
 16 |\mathcal{S}| \cdot  \|\bm{a}_{\mathcal{S}} \|_{2}^{2}  \max _{i,j\in [p]}|\widehat{\Gamma}_{ ij}-\Gamma_{i j}|. 
\end{aligned}
\end{equation*}
Therefore, 
$$\bm a^{\top} \widehat{\bm\Gamma} \bm a \geq \bm a^{\top} \bm\Gamma \bm a -  16 |\mathcal{S}| \cdot \|\bm{a}_{\mathcal{S}}\|_{2}^{2}  \max _{i,j\in [p]}|\widehat{\Gamma}_{ ij}-\Gamma_{i j}|,$$
and
\begin{equation*}
\begin{aligned}
\inf_{\substack{\|\bm{a}_{\mathcal{S}^c} \|_{1} \leq  3\| \bm{a}_{\mathcal{S}}\|_{1} \\ \| \bm{a}\|_2 \neq 0}} \frac{\bm a^{\top}\widehat{\bm{\Gamma}} \bm a}{\|\bm{a}_{\mathcal{S}}\|_{2}^{2}}
& \geq 
\inf_{\substack{\|\bm{a}_{\mathcal{S}^c} \|_{1} \leq  3\| \bm{a}_{\mathcal{S}}\|_{1} \\ \| \bm{a}\|_2 \neq 0}} \frac{\bm a^{\top}{\bm{\Gamma}} \bm a}{\|\bm{a}_{\mathcal{S}}\|_{2}^{2}} -
16 |\mathcal{S}|   \max _{i,j\in [p]}|\widehat{\Gamma}_{ ij}-\Gamma_{i j}|\\
& \geq c_3 - 16 |\mathcal{S}|  \max _{i,j \in [p]}|\widehat{\Gamma}_{ ij}-\Gamma_{i j}| \geq {c_3}/{2},
\end{aligned}
\end{equation*}
where the last inequality follows from $\max _{i,j\in [p]}|\widehat{\Gamma}_{ ij}-\Gamma_{i j}| = O_p(\sqrt{\log p /n})$ and $|\mathcal{S}|\sqrt{\log p /n} = o(1)$. This completes the proof of Lemma \ref{lemma1}.
\end{proof}

% Proof of Theorem 1

\begin{proof}[Proof of Theorem \ref{theorem1} (a)]

Define $\widehat{\bm\Phi}^{\star}=\bm\Phi \widehat{\mathbf M}^\top$ and $\bm\Phi^{\star}=\bm\Phi \mathbf M$. We can rewrite $\bm{v}$ as:
\begin{equation}
\label{eq:A.1}
\bm v=\frac{1}{\sqrt{n}} \widehat{\mathbf M} \bm\Phi^{\top} \mathbf W \bm\varepsilon 
= \frac{1}{\sqrt{n}}  \bm\Phi^{\star\top} \mathbf W \bm\varepsilon + \frac{1}{\sqrt{n}}(\widehat{\bm\Phi}^{\star}-\bm\Phi^{\star})^{\top} \mathbf{W} \bm\varepsilon.
\tag{A.1}
\end{equation}
We next show that the first part on the right hand side of (\ref{eq:A.1}) is asymptotically normal and that the second part is dominated by the first part in probability.

Let $\widetilde{\bm v} =   \bm\Phi^{\star\top} \mathbf W \bm\varepsilon/{\sqrt{n}}$ and its $j$-th entry be  $\widetilde{v}_j = \sqrt{n} \sum_{i=1}^n w_i {\bm{m}}_j^\top \bm{\phi}_i  \varepsilon_i$. Under Conditions \ref{cond:1} and \ref{cond:4}, by \cite{stute1996distributional}, we have that $\widetilde{v}_j$ is asymptotically normal:
 $$\widetilde{v}_j = \sqrt{n} \sum_{i=1}^n w_i {\bm{m}}_j^\top \bm{\phi}_i  \varepsilon_i \stackrel{d}{\rightarrow}  \mathcal{N}(0, \Lambda_{jj}), $$
where  $\Lambda_{jj} = \bm{m}_j^\top\mathbf{\Sigma}\bm{m}_j$. Furthermore, the convergence properties for $\widetilde{v}_j$ are uniform for all $j \in[p]$ under the sub-Gaussian error and bounded covariates conditions in Condition \ref{cond:2} by the same arguments in \cite{chai2019inference}.
Consider the second part on the right hand side of (\ref{eq:A.1}). Define $u_j = \sqrt{n} \sum_{i=1}^n w_i(\widehat{\bm{m}}_j -\bm{m}_j)^\top \bm{\phi}_i \varepsilon_i$. Following similar arguments as Lemma A.20 in \cite{dai2023scale} and \cite{stute1996distributional}, under the conditions of Theorem \ref{theorem1}, we can derive that $\max_{j\in[p]} \|\widehat{\bm{m}}_j -\bm{m}_j\|_2 = O_p(\sqrt{|\mathcal{S}| \log p/n})$. Therefore, by Slutsky's lemma, we have $u_j = o_p(1)$ under the condition that $|\mathcal{S}|\sqrt{\log p/n} = o(1)$. This completes the proof of Theorem \ref{theorem1} (a).
\end{proof}

\begin{proof}[Proof of Theorem \ref{theorem1} (b)]

By the definition of $\bm\Delta$ and $|\widehat{\mathbf{M}} \widehat{\bm\Gamma}-\mathbf I |_{\infty} \leq \mu = O( \sqrt{\log p / n})$, where $|\cdot|_\infty$ is the entry-wise $l_\infty$ norm,  we have:
$$
\|\bm \Delta\|_{\infty} \leq  \sqrt{n} |\widehat{\mathbf{M}} \widehat{\bm\Gamma}-\mathbf I |_{\infty}\|\widehat{\bm{\theta}} - \bm{\theta}_{0}\|_{1} \leq  \sqrt{n} \mu\|\widehat{\bm{\theta}} - \bm{\theta}_{0}\|_{1}.
$$
For the proof of Theorem \ref{theorem1} (b), since $\sqrt{n}\mu\lambda|\mathcal{S}| \rightarrow 0$ as $n \rightarrow \infty$, 
it is sufficient to prove that $\|\widehat{\bm{\theta}} - \bm{\theta}_{0}\|_{1}=O_p(\lambda|\mathcal{S}|)$.

Similar to the proof of Lemma 2 in \cite{chai2019inference}, we have that the event $\mathcal{F}=\{\|\bm\Phi^{\top} \mathbf{W} \bm{\varepsilon} / n\|_{\infty}<{\lambda}/{2}\}$ satisfies $P(\mathcal{F}) \rightarrow 1$ as $n \rightarrow \infty$, and $\max _{i,j\in [p]}|\widehat{\Gamma}_{ ij}-\Gamma_{i j}| = O_p(\sqrt{\log p /n})$.
By the definition of $\widehat{\bm\theta}$, we have:
\begin{equation}
\label{eq:A.2}
 \frac{1}{2n} \| \mathbf{W}^{1/2} (\mathbf{y}-\bm{\Phi} \widehat{\bm\theta}) \|_2^2 -  \frac{1}{2n} \| \mathbf{W}^{1/2} (\mathbf{y}-\bm{\Phi} \bm{\theta}_0) \|_2^2
\leq \lambda\|\bm\theta_{0}\|_{1}-\lambda\|\widehat{\bm\theta}\|_{1},
\tag{A.2}
\end{equation}
and the left-hand side of equation (\ref{eq:A.2}) is: 
$$\text{LHS}=\frac{1}{2 n}\left[(\widehat{\bm\theta}-\bm\theta_{0})^{\top} \bm\Phi^{\top}  \mathbf{W} \bm\Phi (\widehat{\bm\theta}-\bm\theta_{0})-2 \bm\varepsilon^{\top} \mathbf{W} \bm\Phi(\widehat{\bm\theta}-\bm\theta_{0})\right].$$
By the fact that $(\widehat{\bm\theta}-\bm\theta_{0})^{\top} \bm\Phi^{\top} \mathbf{W}  \bm\Phi (\widehat{\bm\theta}-\bm\theta_{0})/n \geq 0$, under event $\mathcal{F}$, we have:
\begin{equation}
\label{eq:A.3}
\begin{aligned}
\frac{1}{n}(\widehat{\bm\theta}-\bm\theta_{0})^{\top} \bm\Phi^{\top}  \mathbf{W}\bm\Phi(\widehat{\bm\theta}-\bm\theta_{0}) & \leq \frac{2}{n} \bm\varepsilon^{\top} \mathbf{W} \bm\Phi(\widehat{\bm\theta}-\bm\theta_{0})+2 \lambda\|\bm\theta_{0}\|_{1}-2 \lambda\|\widehat{\bm{\theta}}\|_{1} \\
& \leq \lambda\|\widehat{\bm\theta}-\bm\theta_{0}\|_{1}+2 \lambda\|\bm\theta_{0,\mathcal{S}}\|_{1}-2 \lambda\|\widehat{\bm\theta}_{\mathcal{S}}\|_{1} -2 \lambda\| \widehat{\bm\theta}_{\mathcal{S}^{c}} -\bm\theta_{0,\mathcal{S}^{c}}\|_{1}\\
& \leq \lambda\|\widehat{\bm\theta}-\bm\theta_{0}\|_{1}+2 \lambda\| \widehat{\bm\theta}_{\mathcal{S}} - \bm\theta_{0,\mathcal{S}}\|_{1} -2 \lambda\|\widehat{\bm\theta}_{\mathcal{S}^{c}}-\bm\theta_{0,\mathcal{S}^{c}}\|_{1}\\
& \leq 3  \lambda\|\widehat{\bm\theta}_{\mathcal{S}}-\bm\theta_{0,\mathcal{S}}\|_{1},
\end{aligned}
\tag{A.3}
\end{equation}
where the second and third inequalities hold due to the fact that $|x|-|y| \leq |x-y|$ for any $x,y \in \mathbb{R}$.
By Lemma \ref{lemma1}, under $\|\widehat{\bm{\theta}}_{\mathcal{S}^c} - \bm{\theta}_{0,\mathcal{S}^c} \|_1 \leq 3 \|\widehat{\bm{\theta}}_{\mathcal{S}} - \bm{\theta}_{0,\mathcal{S}} \|_1$, we have:
\begin{equation}
\label{eq:A.4}
\frac{1}{n}(\widehat{\bm\theta}-\bm\theta_{0})^{\top} \bm\Phi^{\top} \mathbf{W}\bm\Phi(\widehat{\bm\theta}-\bm\theta_{0}) = (\widehat{\bm\theta}-\bm\theta_{0})^{\top} \widehat{\bm{\Gamma}} (\widehat{\bm\theta}-\bm\theta_{0}) \geq \frac{c_3}{2} \| \widehat{\bm{\theta}}_{\mathcal{S}} - \bm{\theta}_{0,\mathcal{S}} \|_{2}^2.
\tag{A.4}
\end{equation}
Combining (\ref{eq:A.3}) and (\ref{eq:A.4}), we have:
$$
\frac{c_3}{2} \| \widehat{\bm{\theta}}_{\mathcal{S}} - \bm{\theta}_{0,\mathcal{S}} \|_{2}^2
\leq 3 \lambda \| \widehat{\bm{\theta}}_{\mathcal{S}} - \bm{\theta}_{0,\mathcal{S}} \|_{1}
\leq 3 \lambda \sqrt{|\mathcal{S}|} \| \widehat{\bm{\theta}}_{\mathcal{S}} - \bm{\theta}_{0,\mathcal{S}} \|_{2},
$$
and obtain that 
$\| \widehat{\bm{\theta}}_{\mathcal{S}} - \bm{\theta}_{0,\mathcal{S}} \|_{2} \leq 6 \lambda \sqrt{|\mathcal{S}|}/c_3$ with probability tending to 1. Therefore, we have
$\| \widehat{\bm{\theta}} - \bm{\theta}_0 \|_1 \leq 4 \| \widehat{\bm{\theta}}_{\mathcal{S}} - \bm{\theta}_{0,\mathcal{S}} \|_{1} = O_p(\lambda |\mathcal{S}|)$. This completes the proof of Theorem \ref{theorem1} (b).

\end{proof}

\begin{proof}[Proof of Theorem \ref{theorem2}]
By the construction of $U_j$, we have:
\begin{equation}
\label{eq:A.5}
U_{j}=\frac{\sqrt{n} \theta_{0 j}}{{\sqrt{\widehat{\Lambda}_{j j}}}}+\frac{v_{j}}{\sqrt{\widehat{\Lambda}_{j j}}}-\frac{\Delta_{j}}{{\sqrt{\widehat{\Lambda}_{j j}}}}.
\tag{A.5}
\end{equation}
From Theorem \ref{theorem1} (b), we have $\|\bm\Delta\|_\infty=o_p(1)$. Under the null hypothesis $\mathcal{H}_j^0: \theta_{0j} = 0$, where the first term of (\ref{eq:A.5}) is zero, if we can show that $v_j/ \widehat{\Lambda}_{jj}^{1/2}\stackrel{d}{\rightarrow} \mathcal{N}(0,1)$, then we have 
$U_j \stackrel{d}{\rightarrow} \mathcal{N}(0,1)$ by Slutsky's lemma. 
By Condition \ref{cond:3} and the definition of $\bm{\Lambda}$, there exists a finite positive constant $c_4$ such that $1/\Lambda_{jj} \leq 1/c_4$. Following the proof of Theorem \ref{theorem1} (a), $\max_{j\in[p]} \| \widehat{\bm{m}}_j -\bm{m}_j\|_2= O_p(\sqrt{|\mathcal{S}| \log p/n})$, and Lemma A.20 in \cite{dai2023scale}, we have $\max_{j\in[p]} |\widehat{\Lambda}_{jj} - {\Lambda}_{jj}| = O_p(\sqrt{|\mathcal{S}| \log p/n})$.
In addition, since
\begin{equation*}
\max_{j\in[p]} |\widehat{\Lambda}_{jj}^{1/2} - {\Lambda}_{jj}^{1/2}| 
\leq \max_{j\in[p]} |\widehat{\Lambda}_{jj} - {\Lambda}_{jj}|/{\Lambda}_{jj}^{1/2}
\leq \max_{j\in[p]} |\widehat{\Lambda}_{jj} - {\Lambda}_{jj}| /\sqrt{c_4},
\end{equation*} 
and following the same arguments as in the proof of Theorem 3.1 in \cite{van2014asymptotically}, we have that $\max_{j\in[p]} |1/\widehat{\Lambda}_{jj}^{1/2} - 1/{\Lambda}_{jj}^{1/2}|=O_p(\sqrt{|\mathcal{S}| \log p/n})$. Therefore, we obtain $\breve{v}_j = v_j/\widehat{\Lambda}_{jj}^{1/2} \stackrel{d}{\rightarrow} \mathcal{N}(0,1)$ by Slutsky's lemma.

Next we show that the overall FDR is controlled at the pre-specified level $\alpha$ using proposed procedure.
We first examine the case that $t_0$ given by (\ref{eq:algorithm1}) does not exist and set $t_0=\sqrt{2\log p}$. Note that
\begin{equation*}
\label{eq:A.6}
\begin{aligned}
P\left(\sum_{j \in \mathcal{S}^{c} \cap [d]} 1\left(\left|U_{j}\right| \leq {\sqrt{2\log p}}\right) \geq 1\right) 
&\leq  
P\left(\sum_{j \in \mathcal{S}^{c} \cap [d]} 1 \left(U_{j} \geq {\sqrt{2\log p}}\right) \geq 1\right)\\
&+P\left(\sum_{j \in \mathcal{S}^{c} \cap [d]} 1 \left(U_{j} \leq -{\sqrt{2\log p}}\right) \geq 1\right),
\end{aligned}
\tag{A.6}
\end{equation*}
and $1/\widehat{\Lambda}_{jj}^{1/2}\leq 1/c_5$ for a positive constant $c_5$ by the fact that $1/\Lambda_{jj}^{1/2} \leq 1/\sqrt{c_4}$ and $\max_{j \in [p]}|1/\widehat{\Lambda}_{jj}^{1/2}-1/{\Lambda}_{jj}^{1/2}|=o_p(1)$.
For any $\epsilon > 0$, the first term on the right hand side of (\ref{eq:A.6}) is bounded by 
\begin{equation*}
%\label{eq:A.8}
\begin{aligned}
&~P\left(\sum_{j \in \mathcal{S}^{c} \cap [d]} 1\left(U_{j} \geq \sqrt{2 \log p}\right) \geq 1\right) \\
= &~ P\left(\sum_{j \in \mathcal{S}^{c} \cap [d]} 1\left(\breve{v}_{j} \geq \sqrt{2 \log p}+ \Delta_j/\widehat{\Lambda}_{jj}^{1/2}\right) \geq 1\right) \\
\leq &~ P\left(\sum_{j \in \mathcal{S}^{c} \cap [d]} 1\left(\breve{v}_{j} \geq \sqrt{2 \log p}- \|\bm{\Delta}\|_\infty/c_5\right) \geq 1\right) \\
\leq &~ P\left(\sum_{j \in \mathcal{S}^{c} \cap [d]} 1\left(\breve{v}_{j} \geq \sqrt{2 \log p}-\epsilon \right) \geq 1\right) + P(\|\bm{\Delta}\|_\infty \geq c_5 \epsilon  )\\
 \leq &~ d\max _{j \in [d]} P\left(\breve{v}_{j} \geq \sqrt{2 \log p} -\epsilon\right)+P\left(\|\bm{\Delta}\|_{\infty} \geq c_5 \epsilon \right)\\
\leq &~ \frac{d}{2} G\left(\sqrt{2 \log p}-\epsilon\right)+P\left(\|\bm\Delta\|_{\infty} \geq c_5 \epsilon \right),
\end{aligned}
%\tag{A.8}
\end{equation*}
which goes to zero as $(n,d)\rightarrow \infty$ by Theorem \ref{theorem1} (b), the fact that $G(t) < 2t^{-1} \phi(t)$ for $\phi(t) = e^{-t^2/2}/\sqrt{2\pi}$ (Lemma 7.1 in \cite{javanmard2019false}),  and the arbitrary of $\epsilon$. By symmetry, the second term on the right hand side of (\ref{eq:A.6}) also goes to zero as $(n,d)\rightarrow \infty$. Similarly, for any $\epsilon >0$, there exist a positive constant $c_6$ such that
\begin{equation*}
\begin{aligned}
&~P\left(\sum_{j \in [d]} \sum_{k \in \mathcal{S}^{c} \cap \mathcal{B}_{2j}} 1\left(U_{j} \geq \sqrt{2 \log p}, U_k \geq \sqrt{2\log p} \right) \geq 1\right) \\
 \leq &~ dq \max _{j \in [d], k\in \mathcal{S}^c \cap \mathcal{B}_{2j}} P\left(\breve{v}_{j} \geq \sqrt{2 \log p} -\epsilon, \breve{v}_{k} \geq \sqrt{2 \log p} -\epsilon\right)+P\left(\|\bm{\Delta}\|_{\infty} \geq c_5 \epsilon \right)\\
 \leq&~ dq c_6 \left(\sqrt{2\log p}-\epsilon + 1\right)^{-2} \exp \left\{-\frac{(\sqrt{2 \log p} - \epsilon)^2}{1+\Lambda_{ij}^0}\right\} + P\left(\|\bm{\Delta}\|_{\infty} \geq c_5 \epsilon \right),
\end{aligned}
\end{equation*}
where the last inequality follows from Lemma 6.2 of \cite{liu2013gaussian} and goes to zero as $(n,d) \rightarrow \infty$ by Theorem \ref{theorem1} (b) and the arbitrariness of $\epsilon$. Therefore, the claim is proved under $t_0=\sqrt{2\log p}$.

Then we consider the case that $t_0$ defined by (\ref{eq:algorithm1}) exists. By the construction of $t_0$, we have $$\frac{dG(t_0)(1+qG(t_0))}{R(t_0) \vee 1} \leq \alpha.$$ Let 
\begin{equation*}
\begin{aligned}
\nu_1& = \sup_{0 \leq t \leq t_{p}}\left|\frac{\sum_{j \in \mathcal{S}^{c} \cap [d] }\left\{ 1\left(\left|U_{j}\right| \geq t\right)-G(t)\right\}}{d G(t)}\right|, ~\text{and}\\
 \nu_2& = \sup_{0 \leq t \leq t_{p}}\left|\frac{\sum_{ j \in[d]}   \sum_{k \in \mathcal{S}^{c} \cap \mathcal{B}_{2j}}\left\{1\left(\left|U_{j}\right| \geq t, \left|U_{k}\right| \geq t \right)-G^2(t)\right\}}{dq G^2(t)}\right|.
\end{aligned}    
\end{equation*}
Then, we can bound FDP$_0$ and $\sum_{j\in[d]}\text{FDP}_j$ as follows:
\begin{equation*}
\label{eq:A.7}
\begin{aligned}
\text{FDP}_0(t_0)&=\frac{\sum_{j \in \mathcal{S}^c\cap [d]} 1\left(\left|U_{j}\right| \geq t_0\right)}{R(t_{0}) \vee 1} \\
&\leq 
\frac{d G\left(t_{0}\right) \nu_1+ \left|\mathcal{S}^c \cap [d]\right| G\left(t_{0}\right)}{R(t_{0}) \vee 1} \\
&\leq 
\frac{d G(t_0)\left(1+\nu_1\right)}{R(t_0) \vee 1},
\end{aligned}
\tag{A.7}
\end{equation*}
and 
\begin{equation*}
\label{eq:A.8}
\begin{aligned}
\sum_{j\in[d]}\text{FDP}_j(t_0)&=\frac{\sum_{j\in[d]}   \sum_{k \in \mathcal{S}^{c} \cap \mathcal{B}_{2j}} 1\left(|U_{j}| \geq t_0, |U_{k}| \geq t_0 \right)}{R(t_{0}) \vee 1} \\
&\leq 
\frac{dq G^2(t_0)\left(1+\nu_2\right)}{R(t_0) \vee 1}.
\end{aligned}
\tag{A.8}
\end{equation*}
Combining (\ref{eq:A.7}) and (\ref{eq:A.8}), we obtain that the overall FDP for the whole procedure is: 
\begin{equation*}
\label{eq:A.9}
\begin{aligned}
\text{FDP}(t_0) \equiv \sum_{j=0}^d \text{FDP}_j(t_0) &\leq \frac{dG(t_0)\left(1+qG(t_0)+\nu_1+qG(t_0)\nu_2\right)}{R(t_0)\vee 1}\\ 
& \leq \alpha + \frac{dG(t_0)(\nu_1+qG(t_0)\nu_2)}{R(t_0)\vee 1}.
\end{aligned}
\tag{A.9}
\end{equation*}
Similar to the proof of Theorem 3.1 and Lemma 6.4 in \cite{liu2013gaussian}, we can obtain that 
\begin{equation}
\label{eq:A.10}
\lim_{(n,d)\rightarrow \infty} \sup_{0\leq t \leq t_p}\left|\frac{\sum_{j\in \mathcal{S}^c \cap [d]}\{P(|U_j| \geq t) - G(t)\}}{d G(t)}\right|= 0,
\tag{A.10}
\end{equation}
and
\begin{equation*}
\label{eq:A.11}
\lim_{(n,d) \rightarrow \infty} \sup_{0\leq t \leq t_p} \left|
\frac{\sum_{j\in [d]}\sum_{k \in \mathcal{S}^{c} \cap \mathcal{B}_{2j}} \{ P \left(|U_j| \geq t, |U_k|  \geq t \right)-G^2(t)\}}{dqG^2(t)} \right|  = 0.
\tag{A.11}
\end{equation*}
By a modification of the Glivenko-Cantelli theorem \cite{berti2002uniform}, it follows that
$$\lim_{(n,d) \rightarrow \infty} \sup_{0\leq t\leq t_p} \left| \frac{1}{d}
\sum_{j\in \mathcal{S}^c \cap [d]} \left\{1(|U_j|\geq t) - P(|U_j| \geq t) \right\}\right| = 0~~ \text{almost surely},$$
and
$$\lim_{(n,d) \rightarrow \infty} \sup_{0\leq t\leq t_p} \left| \frac{1}{dq}
\sum_{j\in[d]} \sum_{k \in \mathcal{S}^c \cap \mathcal{B}_{2j}} \left\{ 
1(|U_j| \geq t, |U_k| \geq t)-P(|U_j| \geq t, |U_k| \geq t)
\right\} 
\right| = 0~~ \text{almost surely}.$$
Therefore, combining with (\ref{eq:A.10}) and (\ref{eq:A.11}), we have
\begin{equation}
\label{eq:A.12}
\begin{aligned}
\limsup \limits_{(n,d) \rightarrow \infty} \nu_1 
\leq & \lim_{(n,d) \rightarrow \infty} \sup_{0< t\leq t_p}
\left| \frac{\sum_{j\in \mathcal{S}^c \cap [d]} \left\{ 1(|U|_j \geq t) - P(|U|_j \geq t) \right\}}{dG(t)}\right|\\
&+
\lim_{(n,d) \rightarrow \infty} \sup_{0\leq t\leq t_p} \left| \frac{\sum_{j\in \mathcal{S}^c \cap [d]} \left\{ P(|U|_j \geq t) - G(t) \right\}}{dG(t)} \right| = 0 ~~ \text{almost surely},
\end{aligned}
\tag{A.12}
\end{equation}
and 
\begin{equation}
\label{eq:A.13}
\begin{aligned}
\limsup \limits_{(n,d) \rightarrow \infty} \nu_2 
\leq & \lim_{(n,d) \rightarrow \infty} \sup_{0 < t\leq t_p}
\left| \frac{\sum_{j\in[d]}\sum_{ k \in {\mathcal{S}^c \cap \mathcal{B}_{2j}}} \left\{ 1(|U|_j \geq t, |U|_k \geq t) - P(|U|_j \geq t, |U|_k \geq t) \right\}}{dqG^2(t)}\right|\\
&+
\lim_{(n,d) \rightarrow \infty} \sup_{0\leq t\leq t_p}
\left| \frac{\sum_{j\in[d]}\sum_{ k \in {\mathcal{S}^c \cap \mathcal{B}_{2j}}} \left\{ P(|U|_j \geq t, |U|_k \geq t) - G^2(t) \right\}}{dqG^2(t)}\right| = 0 ~~ \text{almost surely},
\end{aligned}
\tag{A.13}
\end{equation}
where the inequalities hold by the triangle inequality. Similarly, we have
$$\lim_{(n,d) \rightarrow \infty} \sup_{0\leq t\leq t_p} \frac{1}{d}\left|R(t)-dG(t)-dq G^2(t)\right| = 0 ~~ \text{almost surely}.$$
By the construction of $t_0$, (\ref{eq:A.9}), (\ref{eq:A.12}), (\ref{eq:A.13}), and continuous mapping theorem, it follows that 
\begin{equation*}
    \begin{aligned}
\limsup \limits_{(n,d) \rightarrow \infty} \text{FDP}(t_0) 
&\leq \limsup \limits_{(n,d) \rightarrow \infty}\left\{ \alpha+
\frac{dG(t_0)(\nu_1 + qG(t_0)\nu_2)}{R(t_0)\vee 1}
\right\}\\
&= \limsup \limits_{(n,d) \rightarrow \infty} \left\{  \alpha + \frac{dG(t_0)\nu_1 + dqG^2(t_0) \nu_2}{dG(t_0) + dqG^2(t_0)}
\right\}
= \alpha,
    \end{aligned}
\end{equation*}
in probability 1.
Furthermore, by Fatou's lemma,
$$
\limsup \limits_{(n,d) \rightarrow \infty} \text{FDR}= \limsup \limits_{(n,d) \rightarrow \infty} \left(E\left[\text{FDP}(t_0)\right] \right) \leq E\left[\limsup \limits_{(n,d) \rightarrow \infty} \text{FDP}(t_0) \right] \leq \alpha.
$$
This complete the proof of Theorem \ref{theorem2}.

\end{proof}

\subsection*{Appendix B: Additional numerical results}
\label{appendixB}

\begin{figure}[h]
\centering
\includegraphics[height=3.8 in, width=5.0 in]{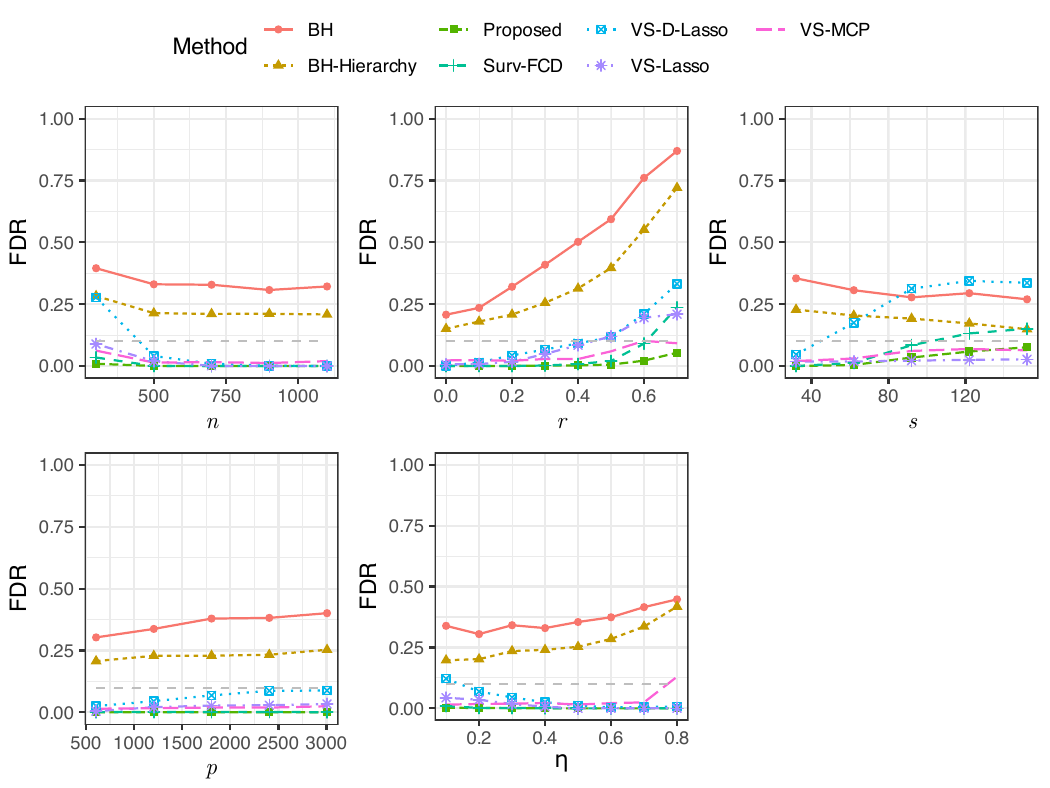}
\caption{
Empirical FDR based on 200 replicates with $\alpha=0.1$ and $a=2$. 
%Dashed grey lines: target FDR; Solid red circles: BH; Dashed orange triangles: BH-Hierarchy; Dashed green squares: proposed; Solid cyan crosses: Surv-FCD; Dashed wathet cross squares: VS-D-Lasso; Dashed purple stars: VS-Lasso; Dashed magenta: VS-MCP. 
Top panels (from left to right) correspond to various sample sizes, censoring rates, and numbers of nonzero covariates, respectively. Bottom panels (from left to right) correspond to different dimensions and correlation coefficients, respectively. Dashed grey lines correspond to target FDR.
}
\label{fig3}
\end{figure}

\begin{figure}[h]
\centering
\includegraphics[height=3.8 in, width=5.0 in]{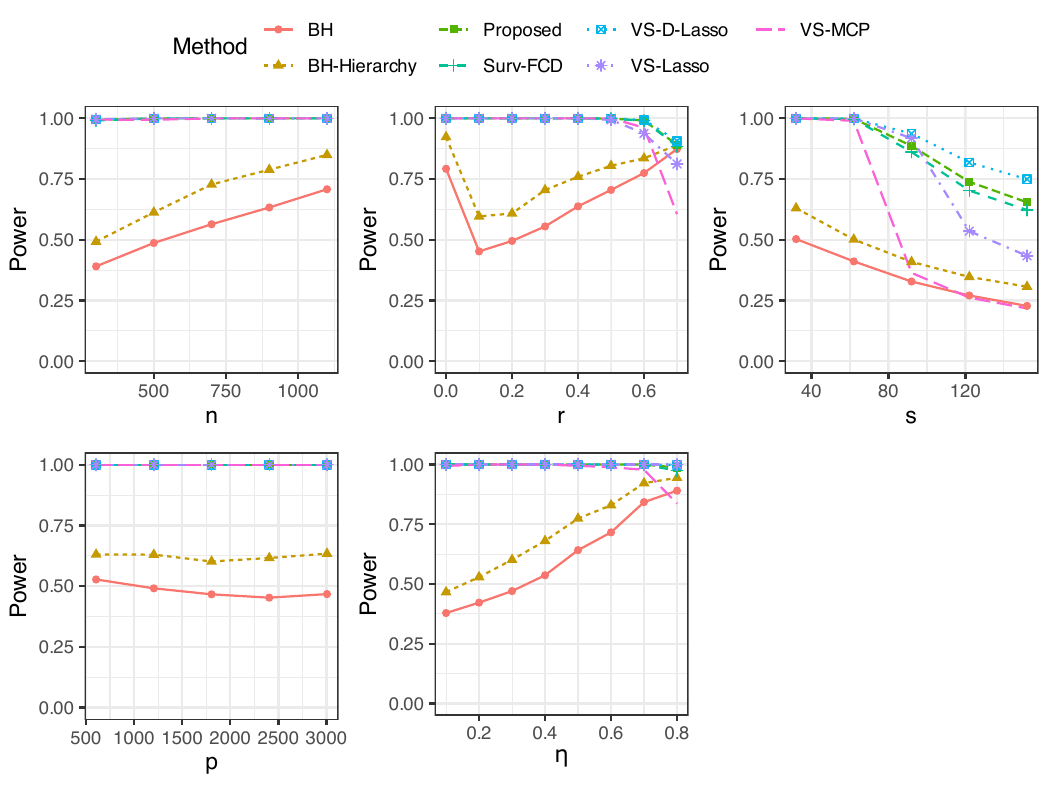}
\caption{Empirical power based on 200 replicates with $\alpha=0.1$ and $a=2$. 
%Dashed grey lines: target FDR; Solid red circles: BH; Dashed orange triangles: BH-Hierarchy; Dashed green squares: proposed; Solid cyan crosses: Surv-FCD; Dashed wathet cross squares: VS-D-Lasso; Dashed purple stars: VS-Lasso; Dashed magenta: VS-MCP. 
Top panels (from left to right) correspond to various sample sizes, censoring rates, and numbers of nonzero covariates, respectively. Bottom panels (from left to right) correspond to different dimensions and correlation coefficients, respectively.}
\label{fig4}
\end{figure}

\clearpage
\begin{figure}[h]
\centering
\includegraphics[height=3.8 in, width=3.5 in]{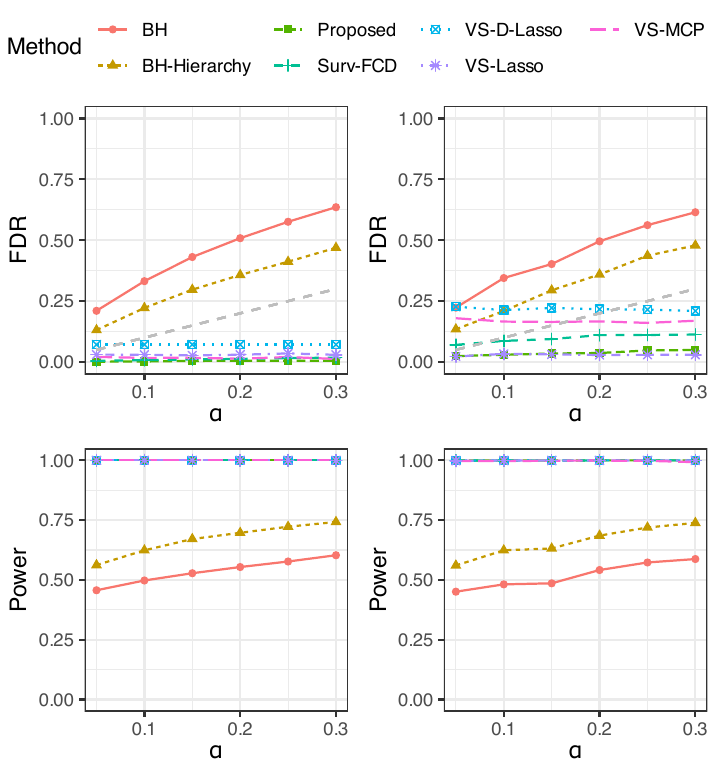}
\caption{Empirical FDR and power based on 200 replicates with different $\alpha$ values. 
Left column corresponds to the exponential survival model, and right column corresponds to the log-logistic survival model. Dashed grey lines correspond to target FDR.
}
\label{fig5}
\end{figure}

%\subsection*{Appendix C: More information on the analysis of breast cancer data}
%\label{appendixC}

{
\noindent{\bf Additional information on the analysis using the proposed approach}
A quick literature search suggests that the genes identified by the proposed method have important implications for breast cancer survival. For example, it was found that the early pathological stages of breast cancer were significantly correlated with the high expression level of COL6A6
\cite{yeh2018extracellular}. As suggested by Table \ref{realdata:1}, with all the other variables equal, an individual with COL6A6 expression one unit greater than another is expected to survive approximately 1.120 $(  =\exp(0.113))$ times longer \citep{klein2014handbook}. According to \cite{wu2021exosomal}, the cooperation between exosomal miR-19a and IBSP leads to the promotion of osteolytic bone metastasis in estrogen receptor-positive breast cancer. 
Interestingly, a significant interaction between IBSP and age at diagnosis was observed to be associated with the progression of breast cancer.
\cite{dustin2019esr1} suggested that the development of ligand-independent ESR1 mutations is a prevalent mechanism for the development of hormonal therapy resistance in metastatic estrogen receptor (ER)-positive breast cancer during aromatase inhibitor therapy.
%Notably, \cite{long2019col6a6} have reported that the interaction between COL6A6 and P4HA3 can suppress the growth and metastasis of pituitary adenoma by blocking PI3K-Akt pathway. This insight motivates us to consider more intricate gene-gene interaction analyses in our future work.
}

\begin{table}[]
\setlength{\abovecaptionskip}{0pt}
\setlength{\belowcaptionskip}{5pt}
\caption{Analysis of the breast cancer data using Surv-FCD: identified main effects and interactions.}
\label{realdata:3}
\centering
\resizebox{210pt}{26mm}
{
\begin{tabular}{ccccc}
\hline
Gene            & Main effects & \multicolumn{3}{c}{ Interactions} \\
\cline{3-5}
                &           & TMB           & Age           & ER Status        \\
                \hline
COL6A6          & 0.113              &               &               &                  \\
CCND2           & 0.085              &               &               &                  \\
ESR1            & -2.322             &               &               & 2.427            \\
NGFR            &                    &               & 0.045         &                  \\
IBSP            &                    &               & 0.044         &                  \\
GNG7            &                    & 0.044         &               &                  \\
IL12A           &                    &               & 0.054         &                  \\
GSTA2           &                    &               &               & -0.022           \\
SLC2A1          &                    &               &               & 0.051            \\
\hline
                &           & \multicolumn{3}{c}{Low-dimensional covariates}                  \\
                \cline{3-5}
                &           & TMB           & Age           & ER Status        \\
\hline
Main effects &           & 0.045         & -0.064        & 4.701       \\
\hline    
\end{tabular}
}
\end{table}

\begin{table}[]
\setlength{\abovecaptionskip}{0pt}
\setlength{\belowcaptionskip}{5pt}
\caption{Analysis of the breast cancer data using VS-D-Lasso: identified main effects and interactions.}
\label{realdata:4}
\centering
\resizebox{210pt}{105mm}
{
\begin{tabular}{ccccc}
\hline
Gene            & Main effects & \multicolumn{3}{c}{ Interactions} \\
\cline{3-5}
                &           & TMB           & Age           & ER Status        \\
                \hline
FGFR4            & 0.027             &               &               &                  \\
INSR             & -0.013            &               &               &                  \\
IGF1R            & -0.023            &               &               &                  \\
IL2RA            & 0.014             &               &               &                  \\
COL6A6           & 0.113             &               &               &                  \\
THBS1            & -0.023            &               &               &                  \\
IBSP             & -0.075            &               & 0.048         &                  \\
LPAR1            & 0.023             &               &               &                  \\
GNG10            & -0.012            &               &               &                  \\
PIK3R5           & -0.011            &               &               &                  \\
AKT3             & -0.019            &               &               &                  \\
MYC              & -0.033            &               &               &                  \\
CDK2             & -0.053            &               &               &                  \\
CCND2            & 0.089             &               &               &                  \\
BCL2             & -0.011            &               &               &                  \\
NR4A1            & -0.043            &               &               &                  \\
AXIN2            & -0.054            &               &               &                  \\
WNT6             & -0.030            & -0.039        &               &                  \\
AGTR1            & -0.012            &               &               &                  \\
ROCK1            & -0.011            &               &               &                  \\
TPM3             & -0.011            &               &               &                  \\
CAMK2B           & 0.022             &               &               &                  \\
CDC42            & -0.046            &               &               &                  \\
PLD2             & 0.024             &               &               &                  \\
ESR1             & -2.322            &               &               & 2.442            \\
RB1              & -0.069            &               &               &                  \\
GADD45G          & 0.034             &               &               &                  \\
CASP8            & -0.014            &               &               &                  \\
GSTM2            & 0.042             &               &               &                  \\
EP300            & -0.076            &               &               &                  \\
DLL3             & -0.038            &               &               &                  \\
RBBP8            &                   &               &               & 0.042            \\
UIMC1            &                   &               &               & -0.013           \\
TGFA             &                   &               &               & -0.026           \\
FGF10            &                   &               &               & 0.023            \\
VEGFA            &                   &               &               & -0.030           \\
VEGFC            &                   &               &               & -0.016           \\
EGFR             &                   &               & 0.014         &                  \\
NGFR             &                   &               & 0.047         &                  \\
NRAS             &                   &               & 0.018         &                  \\
IL6              &                   &               &               & -0.017           \\
GNG7             &                   & 0.044         &               &                  \\
SGK3             &                   &               &               & 0.024            \\
PPP2R5E          &                   &               & -0.021        &                  \\
PPP2R5A          &                   &               &               & 0.013            \\
CCNE1            &                   &               & 0.017         &                  \\
CREB5            &                   & -0.020        &               &                  \\
WNT3A            &                   &               & -0.014        &                  \\
WNT7B            &                   &               &               & -0.015           \\
IL12A            &                   & -0.011        & 0.056         &                  \\
CALML3           &                   &               &               & -0.017           \\
E2F2             &                   &               & 0.031         &                  \\
E2F3             &                   &               &               & -0.049           \\
SMAD2            &                   &               & 0.012         &                  \\
CTBP1            &                   &               & 0.025         &                  \\
GSTA2            &                   &               &               & -0.024           \\
SLC2A1           &                   &               &               & 0.057            \\
\hline
                &           & \multicolumn{3}{c}{Low-dimensional covariates}                  \\
                \cline{3-5}
                &           & TMB           & Age           & ER Status        \\
\hline
Main effects &           & 0.045         & -0.064        & 4.701       \\
\hline    
\end{tabular}
}
\end{table}

\begin{table}[]
\setlength{\abovecaptionskip}{0pt}
\setlength{\belowcaptionskip}{5pt}
\caption{Analysis of the breast cancer data using VS-Lasso: identified main effects and interactions.}
\label{realdata:5}
\centering
\resizebox{210pt}{76mm}
{
\begin{tabular}{ccccc}
\hline
Gene            & Main effects & \multicolumn{3}{c}{ Interactions} \\
\cline{3-5}
                &           & TMB           & Age           & ER Status        \\
                \hline
FGFR4            & 0.019             &               &               &                  \\
IGF1R            & -0.015            &               &               &                  \\
COL6A6           & 0.111             &               &               &                  \\
THBS1            & -0.024            &               &               &                  \\
IBSP             & -0.069            &               & 0.033         &                  \\
GNG10            & -0.033            &               &               &                  \\
PIK3R5           & -0.011            &               &               &                  \\
MYC              & -0.038            &               &               &                  \\
CDK2             & -0.049            &               &               &                  \\
CCND2            & 0.077             &               &               &                  \\
YWHAB            & -0.013            &               &               &                  \\
NR4A1            & -0.046            &               &               &                  \\
AXIN2            & -0.040            &               &               &                  \\
WNT6             & -0.031            & -0.036        &               &                  \\
ROCK1            & -0.011            &               &               &                  \\
CAMK2B           & 0.014             &               &               &                  \\
CDC42            & -0.039            &               &               &                  \\
PLD2             & 0.016             &               &               &                  \\
ESR1             & -2.309            &               &               & 2.384            \\
RB1              & -0.074            &               &               &                  \\
GADD45G          & 0.022             &               &               &                  \\
GSTM2            & 0.042             &               &               &                  \\
EP300            & -0.064            &               &               &                  \\
DLL3             & -0.028            &               &               &                  \\
RBBP8            &                   &               &               & 0.026            \\
TGFA             &                   &               &               & -0.028           \\
FGF10            &                   &               &               & 0.011            \\
VEGFA            &                   &               &               & -0.021           \\
VEGFC            &                   &               &               & -0.013           \\
NGFR             &                   &               & 0.040         &                  \\
NRAS             &                   &               & 0.011         &                  \\
GNG7             &                   & 0.043         &               &                  \\
SGK3             &                   &               &               & 0.011            \\
CCNE1            &                   &               & 0.018         &                  \\
CREB5            &                   & -0.017        &               &                  \\
WNT7B            &                   &               &               & -0.015           \\
IL12A            &                   &               & 0.051         &                  \\
E2F2             &                   &               & 0.017         &                  \\
E2F3             &                   &               &               & -0.024           \\
CTBP1            &                   &               & 0.018         &                  \\
GSTA2            &                   &               &               & -0.019           \\
SLC2A1           &                   &               &               & 0.038            \\
\hline
                &           & \multicolumn{3}{c}{Low-dimensional covariates}                  \\
                \cline{3-5}
                &           & TMB           & Age           & ER Status        \\
\hline
Main effects &           & 0.055          & -0.063         & 4.750        \\
\hline    
\end{tabular}
}
\end{table}

\begin{table}[]
\setlength{\abovecaptionskip}{0pt}
\setlength{\belowcaptionskip}{5pt}
\caption{Analysis of the breast cancer data using VS-MCP: identified main effects and interactions.}
\label{realdata:6}
\centering
\resizebox{210pt}{18mm}
{
\begin{tabular}{ccccc}
\hline
Gene            & Main effects & \multicolumn{3}{c}{ Interactions} \\
\cline{3-5}
                &           & TMB           & Age           & ER Status        \\
                \hline
COL6A6          & 0.080              &               &               &                  \\
IBSP            & -0.019             &               &               &                  \\
ESR1            & -2.591             &               &               & 2.722            \\
TXNRD1          & -0.276             &               &               &                  \\
CCNE1           &                    &               & 0.057         &                  \\
\hline
                &           & \multicolumn{3}{c}{Low-dimensional covariates}                  \\
                \cline{3-5}
                &           & TMB           & Age           & ER Status        \\
\hline
Main effects &          & 0.094         & -0.054        & 4.896       \\
\hline    
\end{tabular}
}
\end{table}

\end{document}